\documentclass[aps,letterpaper,nofootinbib,showpacs,preprintnumbers,amsmath,amssymb]{article}
\usepackage{hyperref}
\usepackage{graphicx}
\usepackage{amsmath}
\usepackage{amsthm,amssymb,amstext,amscd,amsfonts,amsbsy,amsxtra,latexsym}
\usepackage{wrapfig}
\usepackage{color}
\usepackage{pdflscape}
\usepackage{lscape}
\usepackage{tabularx}
\usepackage{rotating} 
\usepackage{pifont}
\usepackage{bbding}
\usepackage{upgreek}
\usepackage{mathrsfs}
\usepackage{wasysym}
\usepackage{verbatim}
\usepackage{subfig}
\usepackage{graphicx}
\usepackage[framemethod=tikz]{mdframed}
\usepackage[paper=a4paper,left=35mm,right=35mm,top=25mm,bottom=25mm]{geometry}

\newcounter{mnotecount}[section]

\renewcommand{\themnotecount}{\thesection.\arabic{mnotecount}}

\newcommand{\mnote}[1]
{\protect{\stepcounter{mnotecount}}$^{\mbox{\footnotesize $%
\!\!\!\!\!\!\,\bullet$\themnotecount}}$ \marginpar{
\raggedright\tiny\em $\!\!\!\!\!\!\,\bullet$\themnotecount: #1} }

\newcommand{\lb}{\label}
\newcommand{\be}{\begin{equation}}
\newcommand{\ee}{\end{equation}}
\newcommand{\ben}{\begin{eqnarray*}}
\newcommand{\een}{\end{eqnarray*}}
\newcommand{\bea}{\begin{eqnarray}}
\newcommand{\eea}{\end{eqnarray}}
\newcommand{\md}{{\mathrm{d}}}
\newcommand{\beq}{\begin{equation}}
\newcommand{\eeq}{\end{equation}}
\newcommand{\beqn}{\begin{equation}\nonumber}
\newcommand{\bean}{\begin{eqnarray}\nonumber}
\DeclareMathAlphabet{\mathpzc}{OT1}{pzc}{m}{it}

\def\cE{{\mathcal E}}
\def\cR{{\mathcal R}}
\def\cH{{\mathcal H}}

\def\cQ{{\mathcal Q}}

\def\cL{{\mathcal L}}
\def\cM{{\mathcal M}}

\def\cC{{\mathcal C}}

\def\cI{{\mathcal I}}
\def\cJ{{\mathcal J}}

\def\dV{{\mbox{dVol}}}
\def\cM{{\mathcal M}}
\def\cRH{{\mathcal R}{\mathcal H}}

\newtheorem{thm}{Theorem}[section]

\newtheorem{prop}[thm]{Proposition}

\renewenvironment{proof}[1][Proof]{\begin{trivlist}
\item[\hskip \labelsep {\bfseries #1:}]}{\qed\end{trivlist}}

\newcommand{\sig}[3]{\Sigma_{#1}^{#2,#3}}

\def\XXint#1#2#3{{\setbox0=\hbox{$#1{#2#3}{\int}$ }
\vcenter{\hbox{$#2#3$ }}\kern-.6\wd0}}

\def\thesection{\arabic{section}}

\makeatletter
\def\p@subsection{}
\def\p@subsubsection{}
\def\p@paragraph{}
\makeatother

\begin{document}

\begin{center}

{\bf {\Large 
Energy Decay in $1+1$ Rindler Spacetime}}\\

\bigskip
\bigskip
Anne T. Franzen\footnote{e-mail address: anne.franzen@tecnico.ulisboa.pt}{${}^{,\star}$} and
Yafet {{E}.} Sanchez Sanchez\footnote{e-mail addresses:yafet.erasmo.sanchez.sanchez@edu.unige.it}{${}^{,\dagger}$}

\bigskip
\bigskip
{${}^{\star}$} {\it Center for Mathematical Analysis, Geometry and Dynamical Systems,}\\
{\it Mathematics Department, Instituto Superior T\'ecnico,}\\ 
{\it Universidade de Lisboa, Portugal}\\
{${}^{\dagger}$}

{\it INFN Sezione di Genova} 
{\it Università di Genova, Italy}

\end{center}
\medskip

\centerline{ABSTRACT}

\noindent
We consider solutions of the massless scalar wave equation $\Box_g\psi=0$ on a fixed Rindler background and show polynomial decay of the energy flux related to Rindler observers near null infinity and to local observers near the Rindler horizon. The main estimates are obtained via the vector field method using suitable vector fields multipliers which are analogous to the ones used in Schwarzschild spacetime and Minkowski spacetime. Furthermore, we compare the Schwarzschild and Rindler scenarios and discuss the extent to which the principle of equivalence holds.
\medskip

\tableofcontents

\section{Introduction}
\lb{intro}

The Rindler spacetime $(\cM,g)$ is a solution to the vacuum Einstein field equations. A brief introduction to the spacetime is for example given
in \cite{muka, wald, waldgr, socolov, semay}. The problem of analyzing the solution of the scalar wave
equation
\bea
\lb{wave_psi}
\Box_g \psi=0
\eea 
on Rindler backgrounds is intimately related to a better understanding of the equivalence principle \cite{aichelburg} and the Unruh effect \cite{unruh} which is often seen as a flat spacetime proxy for the Hawking radiation close to black hole event horizons. Therefore, solutions to the wave equation in Schwarzschild  spacetime under suitable approximations must resemble solutions to the wave equation in Rindler spacetime.

On one hand, the analysis of solutions to the massless scalar wave equations in the exterior region of  Schwarzschild spacetime was a first step towards understanding the non-linear behaviour of Einstein's equations.  Kay and Wald were the first to prove uniform boundedness of solutions up to and including the event horizon \cite{boundnesskay}.  Later on,  Dafermos and Rodnianski introduced robust physical space based methods in order to obtain non-degenerate energy and pointwise decay estimates in the exterior region incluiding the event horizon and future null infinity \cite{redshift, rp, lectures}. In the last decades  immense progress has been made  in the analysis of linear and non-linear wave equations in the exterior of Schwarzschild \cite{klai, semilinear, blue,  holl, hung, hintz, nonlinear}. 

On the other hand, massless scalar wave equations in Rindler spacetime has been mainly analysed from the perspective of  quantum field theory (see the review  \cite{unruhreview, unruhreview2} and references therein).  The analysis in these cases  is done via  mode decomposition. These decompositions has been useful to construct physical quantum states in the presence of Killing horizons \cite{fulling2, kaywald} and has led to general constructions based on microlocal techniques \cite{rad, wrochna, junker, adiabatic}. Furthermore,  Higuchi, Iso, Ueda and Yamamoto have  shown how different modes entangle and their role in the origin of the  radiation \cite{entanglement}. 

{
In this article,  based on the Schwarzschild case, we define a future directed timelike vector field which captures the energy flux associated to local observers at the Rindler horizon and   the energy flux associated to  accelerated observers  in the far away region. Furthermore, we show arbitrary polynomial energy decay of solutions with compactly supported initial data in $1+1$ Rindler spacetime (Theorem \ref{main}). 

}



\subsection{Outline of the paper}
In Section 2, we introduce the necessary preliminaries. We review the vector field method, introduce the coordinates and establish the general notation.  In Section 3, we state the main vector fields used for the proof of energy decay. In Section 4, we show several estimates related to different region of the spacetime and show the main theorem (Theorem \ref{main}). {In Section 5, we discuss the similarities and differences of our results compared to the Schwarzschild scenario.
}

\section{Preliminaries}

\subsection{Vector field method and energy currents}
\lb{vectorfieldmethod}

In the following section we will briefly review the vector field method which we are going to use as an essential tool throughout this work.
The solutions of the  wave equation \eqref{wave_psi}  have an associated  symmetric stress-energy tensor  given by
\be
\lb{energymomentum}
T_{\mu\nu}(\psi)=\partial_\mu\psi\partial_\nu\psi-\frac12g_{\mu\nu}g^{\alpha\beta}\partial_\alpha\psi
\partial_\beta\psi.
\ee
Since $\psi$ is a solution to \eqref{wave_psi} it follows
\begin{equation}
\label{divfree}
\nabla^\mu T_{\mu\nu}=0.
\end{equation}

By contracting the energy-momentum tensor with a vector field $V$, we define the current
\be
\lb{J}
J_\mu^V({\psi})\doteq T_{\mu\nu}({\psi}) V^\nu.
\ee
If the vector field $V$ is timelike, then the one-form $J_\mu^V$ can be interpreted as an energy flux.
In combination with the divergence theorem for a spacetime region $\mathcal{D}$ which is bound by two homologous hypersurfaces, $\Sigma_{\tau}$ and $\Sigma_0$, we have
\begin{equation}
\label{divthe}
\int_{\Sigma_{\tau}} J^V_\mu ({\psi})n^\mu_{\Sigma_{\tau}} \dV_{\Sigma_{\tau}}
+\int_{\mathcal{B}} \nabla^\mu J_\mu({\psi}) \dV=
\int_{\Sigma_0} J^V_\mu({\psi}) n^\mu_{\Sigma_0} \dV_{\Sigma_0}.
\end{equation}
The vector $n^{\mu}_{\Sigma}$ denotes the normal to the subscript hypersurface $\Sigma$ oriented according to Lorentzian geometry convention.
 Further, $\dV$ denotes the volume element over the entire spacetime region and $\dV_{\Sigma}$ the volume elements on $\Sigma$, respectively.
We have that the divergence of the current \eqref{J} is given by
\be
\lb{divergenc}
\nabla^{\mu}J_{\mu}=\nabla^{\mu}(T_{\mu\nu}V^{\nu})=T_{\mu\nu}(\nabla^{\mu} V^{\nu})+(\nabla^{\mu}T_{\mu\nu})V^{\nu}.
\ee
We decompose the divergence into two terms,
\be
\lb{K}
K^V({\psi})\doteq T({\psi})(\nabla V)=(\pi^V)^{\mu\nu}T_{\mu\nu}({\psi}),
\ee
where $(\pi^V)^{\mu\nu} \doteq \frac{1}{2} (\cL_V g)^{\mu\nu}$ is the so called deformation tensor
along $V$,
and
\be
\lb{E}
\cE^V({\psi})\doteq (\nabla^{\mu}T_{\mu\nu})V^{\nu}=(\Box_g {\psi})V({\psi}).
\ee
Thus
\be
\lb{diver}
\nabla^{\mu}J_{\mu}=K^V+\cE^V.
\ee
From \eqref{K} we see that $K^V$ is zero in case that our multiplier is a Killing vector field and $\cE^V$ is zero if ${\psi}$ is a solution to the homogeneous wave equation


\subsection{The Rindler solution}
\lb{rindler_intro}
In the following, we will briefly recall Rindler spacetime, a solution to the Einstein vacuum field equations represented as 
\begin{equation}
    R_{\mu\nu} = 0,
\end{equation}
where \( R_{\mu \nu} \) is the Ricci tensor. This spacetime solution is crucial for understanding a family of uniformly accelerated observers in flat spacetime, and it is noteworthy that the Rindler spacetimes are isometric to a subset of the Minkowski spacetime. Uniform acceleration in relativity is defined as constant proper acceleration, which refers to acceleration as experienced from the observer's frame of reference. For those unfamiliar with Rindler spacetime, references such as \cite{muka,socolov} provide a concise overview.

\subsubsection{The metric, ambient differential structure and Killing vector fields}
\lb{structure}
To set the semantic convention, whenever we refer to the Rindler solution $(\cM, g)$ we mean the right wedge \mbox{$\cR= \cJ^-(\cR\cH_A^+)\cap \cJ^+(\cR\cH_A^-)$} of the spacetime, whose maximally analytic extension is the $1+1$ Minkowski spacetime, see Figure \ref{rindler}. Since this paper deals exclusively with the right wedge, we will for simplicity in the following refer to $\cR\cH_A$ just by using $\cR\cH$. 
{\begin{figure}[!ht]
\centering
\includegraphics[width=0.4\textwidth]{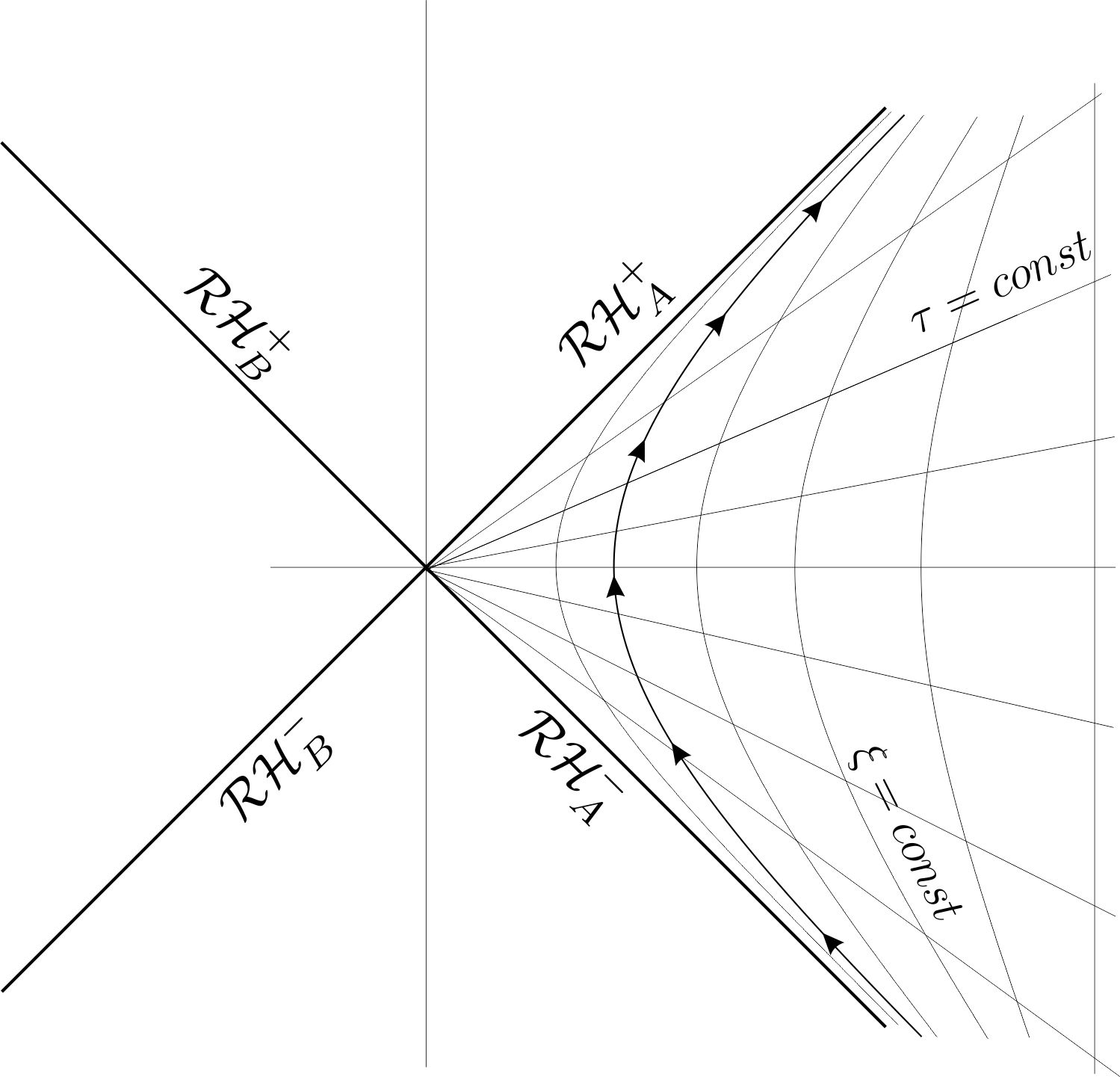}
\caption[]{Sketch of the Rindler spacetime with maximal analytic extension shown}
\label{rindler}\end{figure}}
The manifold $\cM$ can be expressed by 
\bea
\lb{MII}
\cM=\cQ\cup\cR\cH^+\cup\cR\cH^-  , \quad \mbox{with} \quad \cQ=(-\infty, \infty)\times(-\infty, \infty),
\eea
 with the retarded and advanced double-null coordinates $u,v \in (-\infty, \infty)$, which we will introduce in Section \ref{dn_section}.
We can formally parametrize the future and past Rindler horizon by
\bea
\lb{horizon_form}
\cR\cH^+&=&\left\{\infty\right\}\times (-\infty, \infty),\\
\cR\cH^-&=&(-\infty, \infty)\times \left\{-\infty\right\},
\eea
as depicted in Figure \ref{rindler}. Notice that in Minkowski double null coordinates $U,V$ (see {Appendix A}) we have 
\bea
\lb{horizon_form2}
\cR\cH^+&=&\left\{0\right\}\times [0, \infty),\\
\cR\cH^-&=&(-\infty, 0]\times \left\{0\right\},
\eea

\subsubsection{Proper coordinates of uniformly accelerated motion}
In the following we will derive Rindler coordinates simply by defining the trajectory of an uniformly accelerated observer in Minkowski spacetime with the metric
\bea
\lb{Min_metric}
\md s^2=-\md t^2+ \md x^2.
\eea
We refer to the $(t,x)$-coordinates as the laboratory frame. In the following we are interested in the observer's proper frame which moves with the observer and is given in  $(\tau,\xi)$-coordinates. Without loss of generality we assume that our observer moves in $x$-direction. 
Using inverse Lorentz transformations and the comoving frame, which is the inertial frame for a given time at which the observer is instantaneously at rest, Mukhanov and Winitzki, \cite{muka} derive the following relations 
\bea
\lb{t-coord}
t(\tau, \tilde{\xi})&=\frac{1+a\tilde{\xi}}{a} \sinh (a\tau), &\qquad -\infty<t<\infty,\\
\lb{x-coord}
x(\tau, \tilde{\xi})&=\frac{1+a\tilde{\xi}}{a} \cosh (a\tau), &\qquad \quad 0<x<\infty,
\eea
which cover only the subdomain $x>|t|$ of Minkowski space, with the horizon $\cR\cH^+$ corresponding to $t=x$ and $\cR\cH^-$ corresponding to $t=-x$, see Figure \ref{rindler}.  
Further, we have
\bea
\lb{tau-coord}
\tau(t, x)&=\frac{1}{2a} \ln{\frac{x+t}{x-t}}, \qquad \quad &\qquad -\infty<\tau<\infty\\
\lb{tilde_xi-coord}
\tilde{\xi}(t, x)&=-\frac{1}{a} +\sqrt{x^2-t^2}, &\qquad -\frac{1}{a}<\tilde{\xi}<\infty.
\eea
In these coordinates, the horizon $\cR\cH$ corresponds to $\tilde{\xi}=-\frac{1}{a}$. The interpretation is that the accelerated observer cannot measure distances longer than $\frac{1}{a}$ in the opposite direction of their acceleration and can therefore not receive signals from beyond $\cR\cH$.
Relations \eqref{t-coord} and \eqref{x-coord} lead to the Rindler metric 
\bea
\lb{metric_erst}
\md s^2=-(1+a\tilde{\xi})^2\md \tau^2+ \md \tilde{\xi}^2.
\eea 

Using the transformation
\bea
\lb{xi_tilde}
\xi=\frac{1+a\tilde{\xi}}{a}, \qquad 0<{\xi}<\infty,
\eea 
with ${\xi}=0$ corresponding to $\cR\cH$, 
we obtain the simplified metric
\bea
\lb{metric}
\md s^2=-a^2\xi^2\md \tau^2+ \md \xi^2.
\eea

\subsubsection{Double-null coordinates}
\lb{dn_section}
In double null coordinates $(u,v)$
\bea
\lb{u-v-coords}
u(\tau, \xi)&=\tau-\frac{1}{a}\ln\xi, &\qquad -\infty<u<\infty,\\
v(\tau, \xi)&=\tau+\frac{1}{a}\ln\xi, &\qquad -\infty<v<\infty,
\eea
where $u=\infty$ corresponds to $\cR\cH^+$ and $v=-\infty$ corresponds to $\cR\cH^-$. Further, we have\footnote{Note, that in Rindler spacetime, the $r$ coordinate is analogous to the familiar Tortoise coordinate, usually denoted by $r_{\star}$, of for example Schwarzschild spacetime and therefore also takes negative values.}
\bea
\lb{r-t-coords}
r(u,v)&=&\frac{v-u}{2}=\frac{1}{a}\ln\xi, \qquad -\infty<r<\infty,\\
\tau(u,v)&=&\frac{v+u}{2},  \qquad -\infty<\tau<\infty.
\eea
The metric \eqref{metric} then transforms into
\bea
\lb{doublmetric}
\md s^2=-e^{a(v-u)}\md u \md v.
\eea
To better understand the geometric meaning of the coordinates $r$ and $\tau$, the reader may refer for example to Figure \ref{bounded_energy}. 

\subsubsection{$(r, \tau)$-coordinates}
Note that using Section \ref{dn_section} we have
\bea
u(r, \tau)=\tau-r, \qquad v(r, \tau)=r+\tau,
\eea
which leads to the metric in 
in $(r, \tau)$-coordinates 
\bea
\lb{r-tau-metric}
\md s^2=e^{2ar}\left[ -\md \tau^2+ \md r^2\right].
\eea

\subsection{Notation}
\lb{nota}
In the following we will use the notations given below, for the level sets of the coordinate functions $u,v,r$ and $\tau$:
\begin{eqnarray*}
\label{levelSets}
\cC_{v_1}=\{p\in\cM \;|\;v=v_1\}\;, \\
{\cal C}_{u_1}=\{p\in\cM \;|\;u=u_1\}\;, \\
{\cal S}_{r}=\{p\in\cM \;|\;r=r_1\}\;, \\
\Sigma_{\tau_1}=\{p\in\cM \;|\;\tau=\tau_1\} \;.
\end{eqnarray*}
We will also write
$$\cC_{v}(u_1,u_2)=\{p\in\cM \;|\; v(p)=v_1\text{ and } u_1\leq u(p)\leq u_2\}\;.$$

\begin{itemize}
 \item The volume element associated to $g$ in $(u, v)$-coordinates is given by 
\bea
\lb{vol_ele}
\dV=\frac{e^{a(v-u)}}{2}\md u\md v, 
 \eea

\item On the null components $\cC_v$ and $\cC_u$, there is no natural choice of normal vector or volume form. So choosing for example $n_{\cC_{v}}$ to be any
future directed vector orthogonal to $\cC_v$, then $\dV_{\cC_{v}}$ is completely determined by Stokes' Theorem; for instance we have the following choices 
\bea
\lb{null_normal_u}
&n^\mu_{\cC_u}=2e^{-a(v-u)}\partial_v, \quad \mbox{with} \quad &\dV_{\cC_u}=\frac{e^{a(v-u)}}{2}\md v\\
\lb{null_normal_v}
&n^\mu_{\cC_v}=2e^{-a(v-u)}\partial_u, \quad \mbox{with} \quad & \dV_{\cC_v}=\frac{e^{a(v-u)}}{2}\md u,
\eea
%
%
%
%
\item Further, we have for the future timelike normal vector on a constant $\tau$ hypersurface and the outward pointing spacelikes normal vectors on a constant $r$ hypersurface
\bea
\lb{null_normal_tau}
&n^\mu_{\tau=c}=e^{-\frac{a(v-u)}{2}}\left(\partial_u+\partial_v\right), \quad \mbox{with} \quad & \dV_{\tau}=\frac{e^{\frac{a(v-u)}{2}}}{2}(\md u+\md v),\\
\lb{normal_r}
& n^\mu_{r=c}=\pm e^{-\frac{a(v-u)}{2}}\partial_{r}, \quad \mbox{with} \quad &\dV_r=\frac{{e^{+\frac{a(v-u)}{2}}}}{2}(\md v-\md u),
\eea

\item Using regular double null coordinates {(Appendix \ref{reg_dn})} we define the normal on  $\cR\cH^+$ and the volume element by 
\bea
\lb{normalhorizon}
&n^\mu_{\cR\cH^+}=\partial_V, \quad \mbox{with} \quad & \dV_{\cR\cH^+}=dV,
\eea
and similarly by taking the usual Minkowski coordinates  we define via a limiting process the normal vector  and volume form in $\cI^+$ by
\bea
\lb{normalinfinity}
&n^\mu_{\cI^+}=\partial_U, \quad \mbox{with} \quad & \dV_{\cI^+}=dU,
\eea
\end{itemize}
Now we define the foliations which we  will use to bound certain regions where we will apply the divergence theorem.
{For arbitrary $r^*<\tilde{r}\in \mathbb{R}$} we define \bea\Sigma_{\tau_0}^{r^*,\tilde{r}}:= \begin{cases} 
      \Sigma_{\tau=\tau_0} & r^*<r<\tilde{r} \\
       \cC_{v=\tau_0-r^*} &  r\le r^* \\
       \cC_{u=\tau_0+\tilde{r}} & \tilde{r}\ge r
   \end{cases}\quad,
\eea
and 
\bea \Sigma_{\tau_0}^{r^*}:=\begin{cases} 
       \Sigma_{\tau=\tau_0} & r\ge r^*\\
       \cC_{v={\tau_0}-r_0} &  r<r^* 
        \end{cases}\quad.
\eea
In these hypersurfaces we have  the normal vectors
\bea n^\mu_{\Sigma_{\tau_0}^{r^*,\tilde{r}}}:= \begin{cases} 
     n^\mu_{ \Sigma_{\tau=\tau_0}} & r^*<r<\tilde{r} \\
       n^\mu_{\cC_{v=\tau_0-r^*}} &  r\le r^* \\
       n^\mu_{\cC_{u=\tau_0+\tilde{r}}} & \tilde{r}\ge r
   \end{cases}\quad,
\eea
and 
\bea n^\mu_{\Sigma_{\tau_0}^{r^*}}:=\begin{cases} 
       n^\mu_{\Sigma_{\tau=\tau_0}} & r\ge r^*\\
       n^\mu_{\cC_{v=\tau_0-r_0}} &  r<r^* 
        \end{cases}\quad.
\eea
Notice that for fixed $r^*,\tilde{r}$  we have ${\cal{M}}=\displaystyle{\cup_{\tau\in \mathbb{R}}}\Sigma_{\tau}^{r^*,\tilde{r}}=\displaystyle{\cup_{\tau\in \mathbb{R}}}\Sigma_{\tau}^{r^*}$
Therefore, we have 
\bea
&\int_{\Sigma_{\tau_0}^{r^*}} J^{N}_\mu n^{\mu}_{\Sigma_{\tau_0}^{r^*}}\dV_{\Sigma_{\tau_0}^{r^*}}= \\
&\int_{\cC_{v_0}(u_0(r^*),\infty)} J^{N}_\mu n^{\mu}_{\cC_{v_0}} \dV_{\cC_{v_0}}+\int_{\Sigma_{\tau_0}\cap\{r^*<r<\tilde{r}\}} J^{N}_\mu n^{\mu}_{\Sigma_{\tau_0}}\dV_{\Sigma_{\tau_0}}+\int_{\cC_{u_0}(v_0(\tilde{r}),\infty)} J^{N}_\mu n^{\mu}_{\cC_{u_0}} \dV_{\cC_{u_0}} \nonumber
\eea
where $u_0(r^*)$ denotes the value of the $u$ variable  at the point of intersection of   $\Sigma_{\tau_0}$ and ${\cal S}_{r^*}$. Similarly,  $v_0(\tilde{r})$ denotes the value of the $v$ variable  at the point of intersection of   $\Sigma_{\tau_0}$ and ${\cal S}_{\tilde{r}}$.

\section{The {V}ector Fields}

\subsection{The Killing vector $\partial_{\tau}$}

The future directed  Killing vector 
{
\bea
\lb{killing_guy}
\partial_{\tau}=a\left[x \partial_t+t \partial_x\right]
\eea
}
{for $a>0$ is associated with the flow of accelerated observers} through spacetime. It  satisfies 
\bea
\lb{T_rel}
\nabla_{\partial_{\tau}} \partial_{\tau}|_{\cR\cH^+}=a\partial_{\tau},
\eea
which defines the surface gravity $a$ at $\cR\cH^+$.
 Furthermore, the quantity $J^{\partial_\tau}_\mu({\psi})$ can be interpreted as the energy flux related to the observer accelerating at a rate $a$. 
Since $\partial_\tau$ is Killing, it follows from \eqref{K} that $K^{\partial_\tau}=0$. So, for solutions of the wave equation on any compact region $\mathcal{D}$, the energy is conserved, i.e.
$$E_{\partial\mathcal{D}}:=\int_{\partial\mathcal{D}} J^{\partial_\tau}_\mu({\psi}) n^\mu_{\partial\mathcal{D}} \dV_{\partial\mathcal{D}}=0,$$
according to \eqref{diver}.

Further, we have for constant $\tau$ hypersurfaces using \eqref{JX_tau}
\bea
\lb{Jtau-tau}
J_{\mu}^{\partial_{\tau}} n^{\mu}_{\Sigma_{\tau}}= \frac{e^{-ar}}{2}\left[(\partial_{\tau} \psi)^2+(\partial_{r} \psi)^2\right],
\eea
and $\dV_{\Sigma_{\tau}}=e^{ar} \md r$.
For  constant $u$ hypersurfaces using \eqref{tensor_uv}, and $\eqref{JX_u}$ we have 
\bea
\lb{Jtau-u}
J_{\mu}^{\partial_{\tau}} n^{\mu}_{\cC_{u}}=e^{-a(v-u)}(\partial_v\psi) ^2,
\eea
and $\dV_{\cC_{u}}=\frac{e^{a(v-u)}}{2}\md v,$.
For constant $v$ hypersurfaces 
\bea
\lb{Jtau-v}
J_{\mu}^{\partial_{\tau}} n^{\mu}_{\cC_{v}}=e^{-a(v-u)}(\partial_u\psi) ^2
\eea
and $\dV_{\cC_{v}}=\frac{e^{a(v-u)}}{2}\md u$.
Moreover, we have the following proposition.
\begin{prop}\lb{energy-conservation1}
Let $r^*,\tilde{r},t_0\in \mathbb{R}$ with $r^*<\tilde{r}$ and  $\psi$ be a solution to \eqref{wave_psi} with compactly supported data \footnote{The compactness of the data  should be consider with respect to regular coordinates at the horizon. In particular, the compactness is not with respect the variable $r$ .} on $\Sigma_{\tau_0}^{r^*,\tilde{r}}$ or in $\Sigma_{\tau_0}$ not necessarily vanishing at $\cR\cH^+$. Then the following statements hold
\bea\lb{boude1}
\int_{\sig{\tau_1}{r^*}{\tilde{r}}} J^{\partial_{\tau}}_\mu n^{\mu}_{\sig{\tau_1}{r^*}{\tilde{r}}}  \le B \int_{\sig{\tau_0}{r^*}{\tilde{r}}} J^{\partial_{\tau}}_\mu n^{\mu}_{\sig{\tau_0}{r^*}{\tilde{r}}},
\eea
\bea\lb{bounde2}
\int_{\Sigma_{\tau_1}^{r^*}} J^{\partial_{\tau}}_\mu n^{\mu}_{\Sigma_{\tau_1}^{r^*}} \le B \int_{\Sigma_{\tau_0}^{r^*}} J^{\partial_{\tau}}_\mu n^{\mu}_{\Sigma_{\tau_0}^{r^*}},
\eea
\bea\lb{bounde3}
\int_{\sig{\tau_0}{r^*}{\tilde{r}}} J^{\partial_{\tau}}_\mu n^{\mu}_{\sig{\tau_0}{r^*}{\tilde{r}}} \le B \int_{\Sigma_{\tau_0}^{r^*}} J^{\partial_{\tau}}_\mu n^{\mu}_{\Sigma_{\tau_0}^{r^*}},
\eea
and
\bea\lb{bounde4}
\int_{\sig{\tau_1}{r^*}{\tilde{r}}} J^{\partial_{\tau}}_\mu n^{\mu}_{\sig{\tau_1}{r^*}{\tilde{r}}} \le B \int_{\Sigma_{\tau_0}} J^{\partial_{\tau}}_\mu n^{\mu}_{\Sigma_{\tau_0}},
\eea
\end{prop}
\begin{proof}
We start by proving Eq. \eqref{boude1}.

Consider the region ${\cal{D}}=\left\{D^+(\sig{\tau_0}{r_0}{\tilde{r}})\cap D^-(\sig{\tau_1}{r_0}{\tilde{r}})\right\}$, where $D^{\pm}$ are the future/past domain of dependence, respectively.  Applying the divergence theorem \footnote{The divergence theorem cannot be directly applied, since the domain is not compact. Therefore, we applied it by taking a sequence of compact regions and taking the limit.} gives
\bea
&\int_{\sig{\tau_1}{r^*}{\tilde{r}}} J^{\partial_{\tau}}_\mu n^\mu_{\sig{\tau_1}{r^*}{\tilde{r}}}\dV_{\sig{\tau_1}{r^*}{\tilde{r}}} +  \int_{\cRH^+} J^{\partial_{\tau}}_\mu n^\mu_{\cRH^+}\dV_{\cRH^+}+\int_{\cI^+} J^{\partial_{\tau}}_\mu n^\mu_{\cI^+}\dV_{\cI^+}\nonumber\\
&=\int_{\sig{\tau_0}{r^*}{\tilde{r}}} J^{\partial_{\tau}}_\mu n^\mu_{\sig{\tau_0}{r^*}{\tilde{r}}}\dV_{\sig{\tau_0}{r^*}{\tilde{r}}} 
\eea
Now because $\cRH^+$ is a null hypersurface $n^\mu_{\cRH^+}$ is null  and therefore 
\ben
 \int_{\cRH^+} J^{\partial_{\tau}}_\mu n^\mu_{\cRH^+}\dV_{\cRH^+}\ge 0.
\een
Similarly, for $\cI^+$ (which is defined  in the limit sense $v\rightarrow \infty$), we have 
\ben
 \int_{\cI^+} J^{\partial_{\tau}}_\mu n^\mu_{\cI^+}\dV_{\cI^+}:=\displaystyle\lim_{v\rightarrow \infty}\int_{\cC_{v}} J^{N}_\mu n^{\mu}_{\cC_{v}} \dV_{\cC_{v}} \ge 0.
\een
Hence,
\bea\label{divergence10}
&\int_{\sig{\tau_1}{r^*}{\tilde{r}}} J^{\partial_{\tau}}_\mu n^\mu_{\sig{\tau_1}{r^*}{\tilde{r}}}\dV_{\sig{\tau_1}{r^*}{\tilde{r}}} \le\int_{\sig{\tau_0}{r^*}{\tilde{r}}} J^{\partial_{\tau}}_\mu n^\mu_{\sig{\tau_0}{r^*}{\tilde{r}}}\dV_{\sig{\tau_0}{r^*}{\tilde{r}}} 
\eea
The proof of Eq.\eqref{bounde2}, Eq.\eqref{bounde3} and Eq.\eqref{bounde4} is analogous. 
\end{proof}


\subsection{The redshift vector field  $Y$ and the local energy vector field $N$}
{
In \cite[Corollary 3.1]{lectures}, Dafermos and Rodnianski construct a redshift vector field $Y$ and vector field $N$  in Schwarzschild spacetime that captures the local observers energy. An equivalent statement holds in Rindler spacetime which we state below. 

First, we define the redshift vector field tailored for the Rindler Horizon. \\

}
{\bf{The vector field  $Y$}}

Let $Y=(1+\delta_1 (e^{a(v-u)}))\hat{Y}+(\delta_2e^{a(v-u)}) \partial_\tau$ where $\hat{Y}=e^{-a(v-u)}\partial_u$.
Then $Y$ satisfies for suitable $\delta_1,\sigma >0$
\begin{enumerate}
\item $Y$ is $\phi_\tau$ invariant; i.e. $[\partial_\tau,Y]=0$
\item $g(Y,Y)|_{\cRH^+}=0, \qquad g(\partial_\tau,Y)|_{\cRH^+}=-\delta_1$
\item $\nabla_Y Y=-\sigma (Y+\partial_\tau)$
\end{enumerate}
Moreover, the following holds on $\cRH^+$
\ben
&\nabla_{\partial_\tau} Y=-aY\\
&\nabla_Y Y=-\sigma(Y+\partial_\tau)\\
\een

{\bf{The vector field  $N$}}

\begin{prop}
\lb{mi}
Let $\psi$ be a solution of $\Box_g \psi=0$. Given $\epsilon > 0$, there
exist $r_0, R \in (-\infty,\infty)$ and a $\tau$-invariant future directed timelike vector field $N$
such that the following holds
\bea
\lb{energy_control}
(a-\epsilon)J_{\mu}^N(\psi) n^{\mu}_{\cC_{v_*}} \leq K^N(\psi)  \quad r\le r_0\label{rr}\\
c J_{\mu}^{\partial_\tau}(\psi) n^{\mu}_{\Sigma_{\tau_*}}\le J_{\mu}^N(\psi) n^{\mu}_{\Sigma_{\tau_*}} \le C J_{\mu}^{\partial_\tau}(\psi) n^{\mu}_{\Sigma_{\tau_*}} \quad r_0\le r\le R \label{int}\\
|K^N|\le\tilde{C} J_{\mu}^{\partial_\tau}(\psi) n^{\mu}_{\Sigma_{\tau_*}} \quad r_0 \le r\le R\label{bulkybulk}\\
J_{\mu}^N(\psi) n^{\mu}_{\cC_{u_*}}=J_{\mu}^{\partial_\tau}(\psi) n^{\mu}_{u=const} \quad R\le r
\eea
for constants $c,C,\tilde{C}$.
\end{prop}
{
\begin{proof}
Let  $N:=Y+\partial_\tau$. As in the Schwarzschild case, one obtains $(a-\epsilon)J_{\mu}^N(\psi) n^{\mu}_{\cC_{v_*}} \leq K^N(\psi)$ on ${\cRH^+}$. By continuity we can extend this estimate to an open neighbourhood which determines $r_0$.  Moreover, by compactness we obtain Eq.\eqref{int} and Eq.\eqref{bulkybulk}. Finally using suitable cutoffs one can extend the vector field smoothly to $\partial_\tau$ for $R\le r$.
\end{proof}
}







\subsection{The vector field $X$}
\lb{X-section}
{
In this section we define a vector field  $X$ which will be useful to obtain Integrated Local Energy Decay estimates. The ansatz of the vector field goes back to Morawetz \cite{mora}.
}

In the following analysis we use $(r, \tau)$-coordinates.
Now we choose a multiplier such that $K^X$ is positive, with a general ansatz
\bea
\lb{Xf}
X_f=f(r)\partial_r.
\eea

Using the above vector field multiplier in \eqref{K_current_rtau} with 
\bea
\lb{function}
f(r)&=& \tanh(ar)+ b\\
f'(r)&=&\frac{a}{\cosh^2(r)},
\eea
with $b$ a constant, we get
\bea
\lb{KXf}
K^{X_f}=\quad \frac{a}{2\cosh^2(ar)}e^{-2ar}\left[(\partial_{\tau} \psi)^2+(\partial_{r} \psi)^2\right]
\eea
by using \eqref{K_current_rtau}.
Using the normal vector \eqref{null_normal_tau} and the multiplier 
\eqref{Xf} we get
\bea
\lb{JXf}
J_{\mu}^{X_f} n^{\mu}_{\Sigma_{\tau}}= -\left(\tanh(ar)+b\right)e^{-ar}(\partial_{\tau} \psi\partial_{r} \psi).
\eea
Using the normal vector \eqref{normal_r} and the multiplier 
\eqref{Xf} we get
\bea
\lb{JXf2}
J_{\mu}^{X_f} n^{\mu}_{\Sigma_{r}}= \left(\tanh(ar)+b\right)\frac{ e^{-ar}}{2}[(\partial_{\tau} \psi)^2+(\partial_{r} \psi)^2].
\eea

\subsection{The vector field $r^p \partial_v$  }
{
Analogous to the $r^p$ estimates of \cite{rp} (see also\cite{mos}) where this vector field  is used to obtain decay in the far away region, we define the vector field multiplier
}
\bea
\lb{Vvectorfield}
V=r^p \partial_v
\eea
and obtain
\bea
\lb{V}
K^{V}=pr^{p-1}e^{-a(v-u)}(\partial_v\psi)^2
\eea
Using the normal vector \eqref{null_normal_u} and the multiplier 
\eqref{Vvectorfield}  we get
\be\lb{JV}     
J_\mu^V n^\mu_{\cC_u}=2 r^{p} e^{-a(v-u)} (\partial_v\psi)^2
\ee
and with \eqref{JX_r}
\bea 
\lb{JV_urp}              
J_\mu^Vn^\mu_r&=&\frac{r^pe^{-ar}}{2}\left[(\partial_\tau\psi)^2+(\partial_r\psi)^2+2(\partial_\tau\psi)(\partial_r\psi)\right]\nonumber \\
&=&\frac{r^pe^{-ar}}{2} [\partial_\tau\psi+\partial_r\psi]^2\nonumber\\
&=&2{r^pe^{-ar}}(\partial_v\psi)^2,
\eea

\section{Energy Decay}

\subsection{Boundedness of the $J^N$-Energy.}


{\begin{figure}[ht]
\centering
\includegraphics[width=0.45 \textwidth]{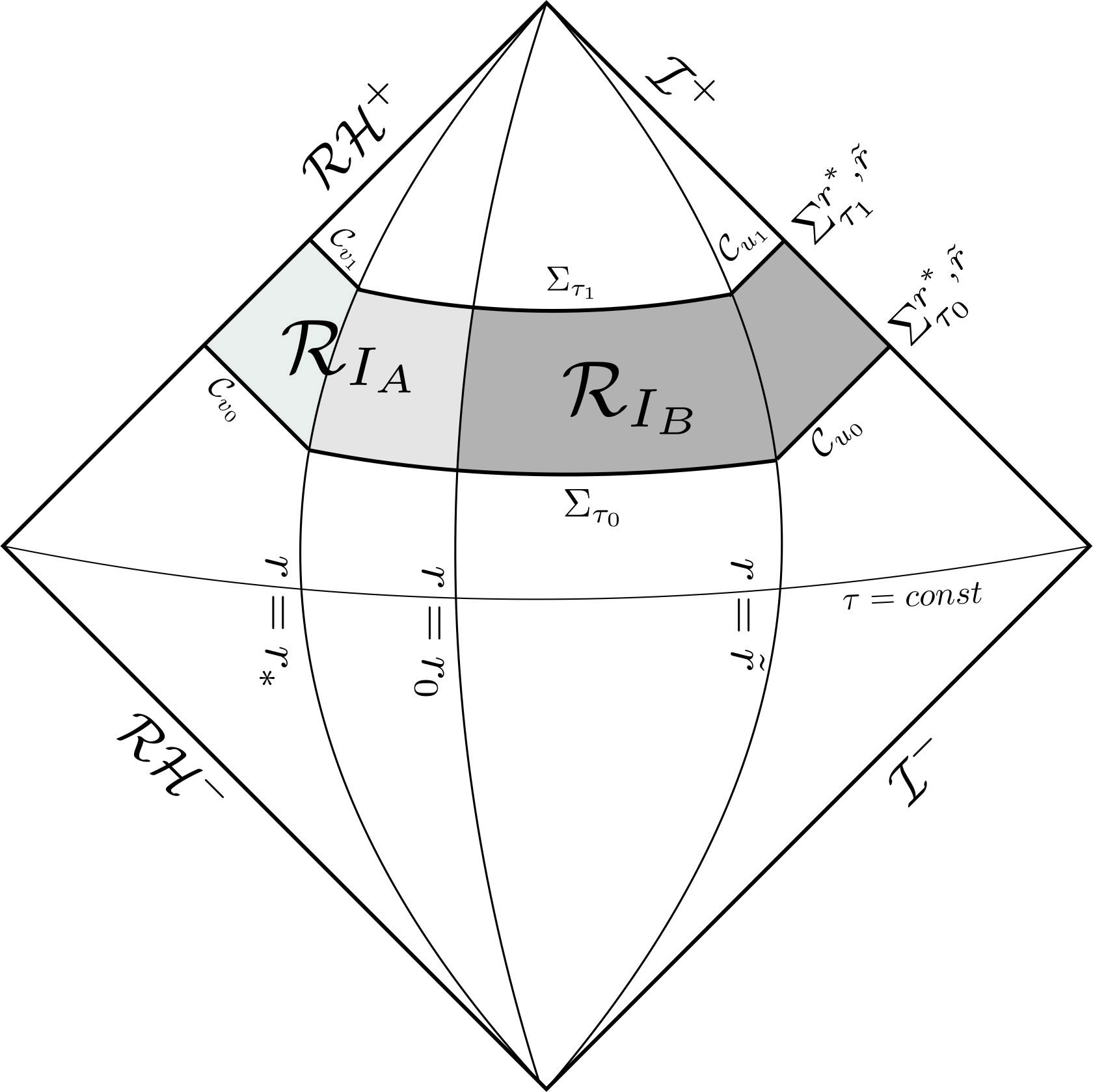}
\caption[]{Sketch of region $\cR_I=R_{I_A} \cup R_{I_B}$ depicted in darker shades}
\label{bounded_energy}\end{figure}}
\begin{prop}\lb{energy-conservation}
Let $r^*,\tilde{r},\tau_0\in \mathbb{R}$ with $r^*<\tilde{r}$ and  $\psi$ be a solution to \eqref{wave_psi} with compactly supported data on $\Sigma_{\tau_0}^{r^*,\tilde{r}}$ or in $\Sigma_{\tau_0}$ not necessarily vanishing at $\cR\cH^+$. Then the following statements hold
\bea\lb{bound1}
\int_{\sig{\tau_1}{r^*}{\tilde{r}}} J^{N}_\mu n^{\mu}_{\sig{\tau_1}{r^*}{\tilde{r}}} \le B \int_{\sig{\tau_0}{r^*}{\tilde{r}}} J^{N}_\mu n^{\mu}_{\sig{\tau_0}{r^*}{\tilde{r}}},
\eea
\bea\lb{bound2}
\int_{\Sigma_{\tau_1}^{r^*}} J^{N}_\mu n^{\mu}_{\Sigma_{\tau_1}^{r^*}} \le B \int_{\Sigma_{\tau_0}^{r^*}} J^{N}_\mu n^{\mu}_{\Sigma_{\tau_0}^{r^*}},
\eea
\bea\lb{bound3}
\int_{\sig{\tau_0}{r^*}{\tilde{r}}} J^{N}_\mu n^{\mu}_{\sig{\tau_0}{r^*}{\tilde{r}}} \le B \int_{\Sigma_{\tau_0}^{r^*}} J^{N}_\mu n^{\mu}_{\Sigma_{\tau_0}^{r^*}},
\eea
and
\bea\lb{bound4}
\int_{\sig{\tau_1}{r^*}{\tilde{r}}} J^{N}_\mu n^{\mu}_{\sig{\tau_1}{r^*}{\tilde{r}}} \le B \int_{\Sigma_{\tau=t_0}} J^{N}_\mu n^{\mu}_{\Sigma_{\tau_0}},
\eea
\end{prop}
\begin{proof}
We start by proving Eq. \eqref{bound1}.
Now consider the region $\cR_I=\left\{D^+(\sig{\tau_0}{r_*}{\tilde{r}})\cap D^-(\sig{\tau_1}{r_*}{\tilde{r}})\right\}={\cR_{I_A} \cup \cR_{I_B}}$, where $D^{\pm}$ are the future/past domain of dependence of the underlying hypersurface, respectively. Further, \mbox{$\cR_{I_A}=\{ p \in {\cR_I} \quad \mbox{and}\quad  r\le r_0, \}$} and \mbox{$\cR_{I_B}=\{ p \in {\cR_I} \quad \mbox{and}\quad  r_0\le r, \}$}, see Figure \ref{bounded_energy}.
Applying the divergence theorem gives
\bea
&\int_{\sig{\tau_1}{r^*}{\tilde{r}}} J^N_\mu n^\mu_{\sig{\tau_1}{r^*}{\tilde{r}}}\dV_{\sig{\tau_1}{r^*}{\tilde{r}}} + \int_{\cR_I} K^N\dV_{\cR_I}\\ \nonumber
&+  \int_{\cRH^+} J^N_\mu n^\mu_{\cRH^+}\dV_{\cRH^+}+\int_{\cI^+} J^N_\mu n^\mu_{\cI^+}\dV_{\cI^+}\nonumber\\ \nonumber
&=\int_{\sig{\tau_0}{r_0}{\tilde{r}}} J^N_\mu n^\mu_{\sig{\tau_0}{r^*}{\tilde{r}}}\dV_{\sig{\tau_0}{r_0}{\tilde{r}}}. 
\eea
Now because $\cRH^+$ is a null hypersurface $n^\mu_{\cRH^+}$ is null  and therefore 
\ben
 \int_{\cRH^+} J^N_\mu n^\mu_{\cRH^+}\dV_{\cRH^+}\ge 0.
\een
Similarly, for $\cI^+$, we have 
\ben
 \int_{\cI^+} J^N_\mu n^\mu_{\cI^+}\dV_{\cI^+}\ge 0.
\een
Hence,
\bea\label{divergence1}
&\int_{\sig{\tau_1}{r^*}{\tilde{r}}} J^N_\mu n^\mu_{\sig{\tau_1}{r^*}{\tilde{r}}}\dV_{\sig{\tau_1}{r^*}{\tilde{r}}} + \int_{\cR} K^N\dV_{\cR}\\ \nonumber
&\le\int_{\sig{\tau_0}{r^*}{\tilde{r}}} J^N_\mu n^\mu_{\sig{\tau_0}{r^*}{\tilde{r}}}\dV_{\sig{\tau_0}{r^*}{\tilde{r}}}. 
\eea
Now consider \mbox{$\cR_{I_A}$} and \mbox{$\cR_{I_B}$}, as shown in Figure \ref{bounded_energy}, and write the bulk integral as
\bea
\int_{{\cR_I}} K^N\dV_{{\cR_I}}=\int_{{\cR_{I_A}}} K^N \dV_{{\cR_{I}}}+\int_{\cR_{I_B}} K^N\dV_{\cR_{I_B}},
\eea
which using \eqref{energy_control} of Proposition \ref{mi} and gives
\bea \label{divergence}
&\int_{{\cR_I}} K^N\dV_{{\cR_I}}=\int_{{\cR_{I_A}}} K^N\dV_{{\cR_{I_A}}}+\int_{\cR_{I_B}} K^N\dV_{\cR_{I_B}}\\ \nonumber
&\ge(a-\epsilon) \int_{{\cR_{I_A}}} J_\mu^N n^{\mu}_{{\sig{\tau}{r^*}{\tilde{r}}\cap\{r\le r_0\}}}\dV_{{\cR_{I_A}}}+\int_{\cR_{I_B}} K^N\dV_{\cR_{I_B}}.
\eea
Moreover, notice that $|K^N|$ in ${\cR_{I_B}}$ is controlled by $J^{\partial_\tau}_\mu n^{\mu}_{\Sigma_{\tau_*}}$   from \eqref{bulkybulk} of Proposition \ref{mi}. 
Therefore using  Proposuition \ref{energy-conservation1} and co-area formula, 
\bea
- \int_{\cR_{I_B}} K^N\dV_{\cR_{I_B}}&\le& \int_{\cR_{I_B}} |K^N|\dV_{\cR_{I_B}}\\
&\le& C  \int_{\cR_{I_B}} J^{\partial_\tau}_\mu n^\mu_{\sig{\tau}{r^*}{\tilde{r}}}\dV_{\cR_{I_B}}\\
&\le& C  \int_{{\cR_I}} J_\mu^{\partial{\tau}}n^\mu_{\sig{\tau}{r^*}{\tilde{r}}}\dV_{{\cR_I}}\\
&\le& C  \int_{\tau_0}^{\tau_1}\int_{\sig{\tau}{r^*}{\tilde{r}}} J_\mu^{\partial{\tau}}n^\mu_{\sig{\tau}{r^*}{\tilde{r}}}\dV_{\sig{\tau}{r^*}{\tilde{r}}}d\tau\\
&\le& C (\tau_1-\tau_0) \int_{\sig{\tau_0}{r^*}{\tilde{r}}} J_\mu^Nn^\mu_{\sig{\tau_0}{r^*}{\tilde{r}}}\dV_{\sig{\tau_0}{r^*}{\tilde{r}}}\label{bulynotsobulk},
\eea
where $C$ can be chosen arbitrarily large.
Hence, combining \eqref{divergence} and \eqref{bulynotsobulk} in \eqref{divergence1} we obtain 
\bea
&[(1+C(\tau_1-\tau_0))]\int_{\sig{\tau_0}{r^*}{\tilde{r}}}J_\mu^Nn^\mu_{\sig{\tau_0}{r^*}{\tilde{r}}} \dV_{\sig{\tau_0}{r^*}{\tilde{r}}}\\ \nonumber
& \ge  (a-\epsilon) \int_{{\cal{R_I}}} J_\mu^N n^\mu_{\sig{\tau}{r^*}{\tilde{r}}}\dV_{{\cal{R_I}}}\\ \nonumber
&+\int_{{\sig{\tau_1}{r^*}{\tilde{r}}}} J^N_\mu n^\mu_{\sig{\tau_1}{r^*}{\tilde{r}}}\dV_{\sig{\tau_1}{r^*}{\tilde{r}}}.
\eea
Adding a term of the form $(a-\epsilon)\int_{\cR_{I_B}} J_\mu^N n^\mu_{\sig{\tau}{r^*}{\tilde{r}}}\dV_{\cR_{I_B}}$ to both sides and using \eqref{bulynotsobulk} we obtain
\bea\label{energyalter}
&[(1+C(\tau_1-\tau_0))]\int_{\sig{\tau_0}{r^*}{\tilde{r}}}J_\mu^Nn^\mu_{\sig{\tau_0}{r^*}{\tilde{r}}} \dV_{\sig{\tau_0}{r^*}{\tilde{r}}}\\ \nonumber
& \ge  (a-\epsilon) \int_{\tau_0}^{\tau_1}\int_{\sig{\tau}{r^*}{\tilde{r}}} J_\mu^N n^\mu_{\sig{\tau}{r^*}{\tilde{r}}}\dV_{\sig{\tau}{r^*}{\tilde{r}}}d\tau\\ \nonumber
&+\int_{{\sig{\tau_1}{r^*}{\tilde{r}}}} J^N_\mu n^\mu_{\sig{\tau_1}{r^*}{\tilde{r}}}\dV_{\sig{\tau_1}{r^*}{\tilde{r}}}.
\eea
Let  $f(\tau)=\int_{\sig{\tau}{r^*}{\tilde{r}}}J_\mu^Nn^\mu_{\sig{\tau}{r^*}{\tilde{r}}} \dV_{\sig{\tau}{r^*}{\tilde{r}}}$, then Eq. \eqref{energyalter} is given by 
\bea
[(1+C(\tau_1-\tau_0))]f(\tau_0)\ge  (a-\epsilon) \int_{\tau_0}^{\tau_1}f(\tau)+f(\tau_1).
\eea
Rearranging terms, diving by $\tau_1-\tau_0$ we have 
\bea
Cf(\tau_0)\ge  \frac{(a-\epsilon)}{\tau_1-\tau_0} \int_{\tau_0}^{\tau_1}f(\tau)+\frac{f(\tau_1)-f(\tau_0)}{\tau_1-\tau_0}.
\eea
Taking the limit $\tau_1\rightarrow \tau_0$ gives
\bea
Cf(\tau_0)\ge  {(a-\epsilon)} f(\tau_1)+\frac{d}{d\tau}f(\tau)|_{\tau_1}.
\eea
This implies
\bea
\frac{d}{d\tau}\left( f(\tau)|_{\tau_1}e^{ {(a-\epsilon)}\tau_1}-\frac{C}{ {(a-\epsilon)} }f(\tau_0)e^{ {(a-\epsilon)} \tau_1}\right)\le 0,
\eea
which gives for some constant $B$
\bea
B \int_{{\sig{\tau_0}{r^*}{\tilde{r}}}} J^N_\mu n^\mu_{\sig{\tau_0}{r^*}{\tilde{r}}}\dV_{\sig{\tau_0}{r^*}{\tilde{r}}}  \ge   \int_{{\sig{\tau_1}{r^*}{\tilde{r}}}} J^N_\mu n^\mu_{\sig{\tau_1}{r^*}{\tilde{r}}}\dV_{\sig{\tau_1}{r^*}{\tilde{r}}}.
\eea
\end{proof}

\subsection{The Intermediate Region}

{\begin{figure}[!ht]
\centering
\includegraphics[width=0.4 \textwidth]{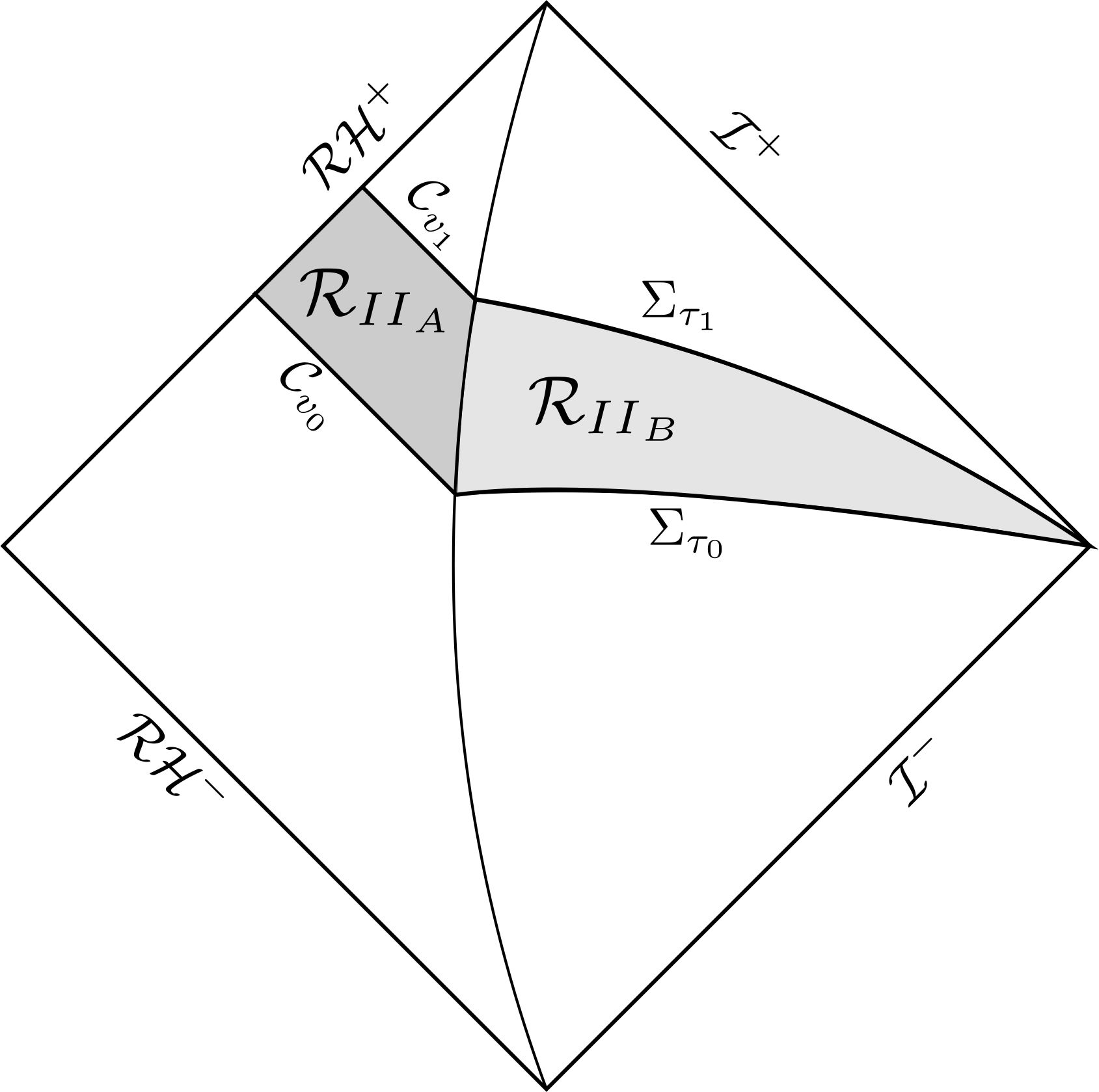}
\caption[]{Sketch of region $R_{II}=R_{{II}_A} \cup R_{{II}_B}$ depicted in darker shades}
\label{intermediate_figure}\end{figure}}
\begin{prop}\lb{decaymedium}
Let $X_f$ be the vector field multiplier \eqref{Xf} given in Section \ref{X-section} and $\psi$ be a solution to \eqref{wave_psi}, the homogeneous wave equation on Rindler with compactly supported data. Then the following statements hold
\bea
\lb{zwei}
C(\tilde{r}) K^{X_f}(\psi) &>& J_{\mu}^{\partial_\tau} n^{\mu}_{\Sigma_{\tau_*}}(\psi), \quad \mbox{for} \quad r\ge \tilde{r}
\eea
and
\bea
\lb{drei}
|J_{\mu}^{X_f} n^{\mu}_{\Sigma_{\tau_*}}(\psi)|&\le& C J_{\mu}^{\partial_\tau} n^{\mu}_{\Sigma_{\tau_*}}(\psi).
\eea
Moreover, we have 
\bea
\lb{funf}
\int\limits_{\cR_{II_B}} J_{\mu}^{\partial_\tau} n^{\mu}_{\Sigma_{\tau_*}}(\psi)\dV_{\cR_{II_B}}&\le& C \int\limits_{\Sigma_{\tau_0}} J_{\mu}^{\partial_\tau} n^{\mu}_{\Sigma_{\tau_*}}(\psi)\dV_{\Sigma_{\tau_0}}, 
\eea
where $\cR_{II}:=\left\{D^+(\Sigma_{\tau_0}^{\tilde{r}})\cap D^{-}(\Sigma_{\tau_1}^{\tilde{r}})\right\}$ and
${\cR_{II_B}}:\{p\in\cR_{II} \in \text{s. t. } r\ge \tilde{r}\} $ for $\tilde{r}\in \mathbb{R}$, see Figure \ref{intermediate_figure}.
 \end{prop}
\begin{proof}
We have 
\bea\label{bla}
K^{X_f}=\quad \frac{a}{2\cosh^2(ar)}e^{-2ar}\left[(\partial_{\tau} \psi)^2+(\partial_{r} \psi)^2\right]
\eea
and
\bea\label{ble}
J_{\mu}^{\partial_{\tau}} n^{\mu}_{\Sigma_{\tau}}= \frac{e^{-ar}}{2}\left[(\partial_{\tau} \psi)^2+(\partial_{r} \psi)^2\right],
\eea
Statement \eqref{zwei} can be shown by comparing the coefficients in \eqref{bla} and \eqref{ble}. Since $\cosh^2(ar)$ is a bounded function, then for $r\in [\tilde{r}, \infty)$ there exists, ${c(\tilde{r})}$,  such that 
\bea
\frac{a}{2\cosh^2(ar)}e^{-ar}>c(\tilde{r})>0.
\eea
Now using the normal vector \eqref{null_normal_tau} and the multiplier 
\eqref{Xf} we get
\bea
\lb{JXf3}
J_{\mu}^{X_f} n^{\mu}_{\Sigma_{\tau}}= -\left(\tanh(ar)+b\right)e^{-ar}(\partial_{\tau} \psi\partial_{r} \psi).
\eea
The proof of statement \eqref{drei} follows from applying Cauchy--Schwarz inequality to \eqref{JXf3} and then comparing with \eqref{ble} to get
\bea
|J_{\mu}^{X_f} n^{\mu}_{\Sigma_{\tau_*}}|&\le&  -\left[\tanh(ar)+b\right]\frac{e^{-ar}}{2}\left[(\partial_{\tau} \psi)^2+(\partial_{r} \psi)^2\right]\nonumber\\
&\le& C e^{-ar}\left[(\partial_{\tau} \psi)^2+(\partial_{r} \psi)^2\right]\nonumber\\
&\le& J_{\mu}^{\partial_\tau} n^{\mu}_{\Sigma_{\tau_*}}.
\eea
In order to show statement \eqref{funf}, the following two steps are applied.
First, we choose $b$ in \eqref{function} such that $\tanh(a\tilde{r})+b=0$. Second, we apply the divergence theorem in $\cR_{II_B}$ using the vector field $X_f$, which give the following inequality
\bea
\lb{compa}
\int\limits_{\cR_{II_B}} K^{X_f}(\psi)\dV_{\cR_{II_B}}\le \left|\int\limits_{\Sigma_{\tau_1}} J_{\mu}^{X_f}(\psi) n^{\mu}_{\Sigma_{\tau_1}}\dV_{\Sigma_{\tau_1}}\right|+\left|\int\limits_{\Sigma_{\tau_0}} J_{\mu}^{X_f}(\psi) n^{\mu}_{\Sigma_{\tau_0}}\dV_{\Sigma_{\tau_0}}\right|.
\eea
The timelike boundary term at $\tilde{r}$ vanishes due to our choice of $b$  and at infinity due to the compact support of $\psi$.
By Eq.\eqref{zwei}, Eq.\eqref{drei} and Proposition \eqref{energy-conservation1}, we have 
\bea
\lb{KXesti}
&\int\limits_{\cR_{II_B}} J_{\mu}^{\partial_\tau}(\psi) n^{\mu}_{\Sigma_{\tau_*}}\dV_{\cR_{II_B}}\le C\int\limits_{\cR_{II_B}} K^{X_f}(\psi)\dV_{\cR_{II_B}}\\
&\le  \left|\int\limits_{\Sigma_{\tau_1}\cap\{r>\tilde{r}\}} J_{\mu}^{X_f}(\psi) n^{\mu}_{\Sigma_{\tau_1}}\dV_{\Sigma_{\tau_1}}\right|+\left|\int\limits_{\Sigma_{\tau_0}\cap\{r>\tilde{r}\}} J_{\mu}^{X_f}(\psi) n^{\mu}_{\Sigma_{\tau_0}}\dV_{\Sigma_{\tau_0}}\right|\\
&\le  \left|\int\limits_{\Sigma_{\tau_1}\cap\{r>\tilde{r}\}} J_{\mu}^{\partial_\tau}(\psi) n^{\mu}_{\Sigma_{\tau_1}}\dV_{\Sigma_{\tau_1}}\right|+\left|\int\limits_{\Sigma_{\tau_0}\cap\{r>\tilde{r}\}} J_{\mu}^{\partial_\tau}(\psi) n^{\mu}_{\Sigma_{\tau_0}}\dV_{\Sigma_{\tau_0}}\right|\\
& \le \int\limits_{\Sigma_{\tau_0}} J_{\mu}^{\partial_\tau}(\psi) n^{\mu}_{\Sigma_{\tau_0}}\dV_{\Sigma_{\tau_0}}.
\eea
\end{proof}


\subsection{Local Energy Decay near $\cRH^+$}
{\begin{figure}[!ht]
\centering
\includegraphics[width=0.45 \textwidth]{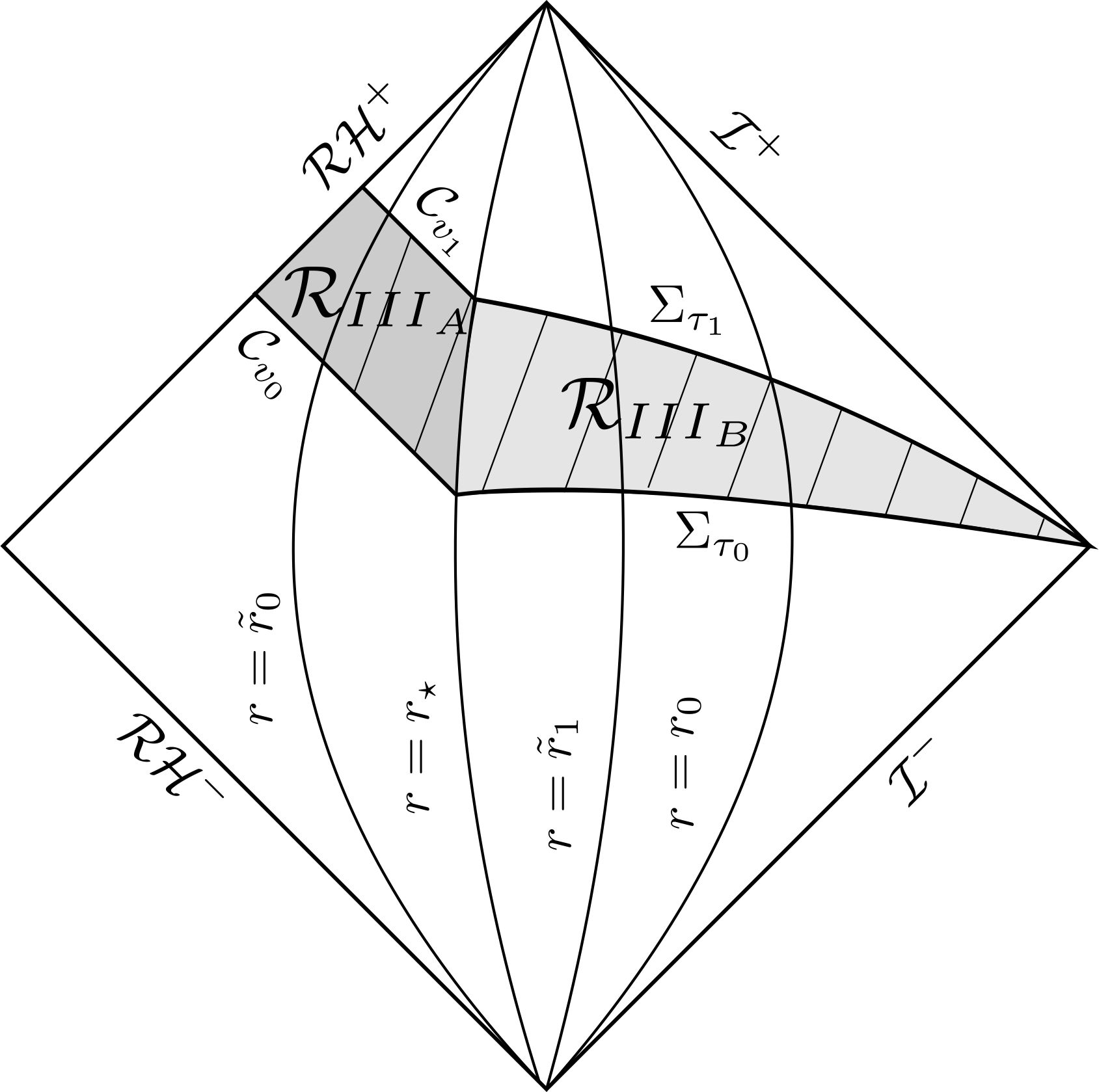}
\caption[]{Sketch of region $R_{III}=R_{{III}_A} \cup R_{{III}_B}$ depicted in darker shades}
\label{localenergy}\end{figure}}
\begin{prop}\lb{decaynear}
Let $N$ be the vector field multiplier given by Proposition \eqref{mi} and $\psi$ be a solution to \eqref{wave_psi} with compactly supported data. Then the following statement hold
\bea
\lb{near}
\int\limits_{\cR_{{III}_A}} J_{\mu}^{N} n^{\mu}_{\cC_{v}}(\psi)\dV_{\cR_{{III}_A}}&\le& C \int\limits_{\Sigma_{\tau_0}} J_{\mu}^{N} n^{\mu}_{\Sigma_{\tau_*}}(\psi)\dV_{\Sigma_{\tau_0}}, 
\eea
where $\cR_{III}:=\left\{D^+(\Sigma_{\tau_0}^{{r}_{\star}})\cap D^{-}(\Sigma_{\tau_1}^{{r}_{\star}})\right\}=R_{{III}_A} \cup R_{{III}_B}$ with
${\cR_{{III}_A}}:\{p\in \cR_{III} \in \text{s. t. } r\le {r}_{\star}\} $ and ${\cR_{{III}_B}}:\{p\in \cR_{III} \in \text{s. t. } {r}_{\star} \le r \} $ for $r_{\star}$ defined as below. See Figure \ref{localenergy}.
 \end{prop}
\begin{proof}
Define  $g(\bar{r})=\int_{\tau_0}^{\tau_1}J^N_{\mu}n^\mu_{\bar{r}}e^{a\bar{r}}\md t $ and choose $\tilde{r}_0,\tilde{r}_1< r_0$ , then from the mean value theorem for integrals  there is a $r_{\star}\in (\tilde{r}_0, \tilde{r}_1)$  such that  $g(r_{\star})=\frac{1}{\tilde{r}_1 -\tilde{r}_0}\int_{\tilde{r}_0}^{\tilde{r}_1}g(\bar{r})\md \bar{r}$.  Take this $r_{\star}$ to define 
$\cR_{III}$.
We use the divergence theorem  in the region $\cR_{{III}_A}$ with the vector field $N$ to obtain
\bea\label{horizon1}
&&\int_{\cC_{v_1}(u_1(r^*),\infty)}J^N_{\mu}n^{\mu}_{\cC_{v}}\dV_{\cC_{v}}+\int_{\cR_{{III}_A}}K^N\dV_{\cR_{{III}_A}}
\nonumber \\ 
&&
\le \int_{\cC_{v_0}(u_0(r^*),\infty)}J^N_{\mu}n^{\mu}_{\cC_{v}}\dV_{\cC_{v}}+\left| \int_{r_{\star}}J^N_{\mu}n^{\mu}_{\cC_{r}}\dV_{\cC_{r}}\right|
\eea
Dropping the positive boundary term on the left hand side and using  \eqref{energy_control} we obtain 
\bea\label{horizon2}
(a-\epsilon)\int_{\cR_{{III}_A}}J^N_\mu n^{\mu}_{\cC_{v}}\dV_{\cR_{{III}_A}}\le \int_{\cC_{v_0}(u_0(r^*),\infty)}J^N_{\mu}n^{\mu}_{\cC_{v}}\dV_{\cC_{v}}+\left| \int_{r_\star}J^N_{\mu}n^{\mu}_{\cC_{r}}\dV_{\cC_{r}}\right|
\eea
To estimate the integral given by a timelike boundary we recall {that by the mean value theorem}  $g(r_{\star})=\frac{1}{\tilde{r}_1 -\tilde{r}_0}\int_{\tilde{r}_0}^{\tilde{r}_1}g(\bar{r})\md \bar{r}$  and we obtain the following inequalities
\begin{align*}
g(r_{\star})&\le |g(r_{\star})|\\
&\le \left|\frac{1}{\tilde{r}_1 -\tilde{r}_0}\int_{\tilde{r}_0}^{\tilde{r}_1}g(\bar{r})\md \bar{r}\right|\\
&\le C \int_{\tilde{r}_0}^{\tilde{r}_1}\int_{\tau_0}^{\tau_1} \left((\partial_r\psi)^2+(\partial_\tau\psi)^2\right) \md\tau \md r\\
&{\le} C \int_{\tilde{r}_0}^{\tilde{r}_1}\int_{\tau_0}^{\tau_1} J_{\mu}^{N}(\psi) n^{\mu}_{\Sigma_{\tau_*}}e^{ar} \md\tau\md r\\
&\le  C\int_{\cR_{{III}_B}}  J_{\mu}^{N}(\psi) n^{\mu}_{\Sigma_{\tau_*}} \dV_{\cR_{{III}_B}}\\
&\stackrel{\eqref{funf}}\le C \int_{ \Sigma_{\tau_0} } J_{\mu}^{N}(\psi) n^{\mu}_{\Sigma_{\tau}}\dV_{ \Sigma_{\tau_0} }
\end{align*}
Putting this together in Eq. \eqref{horizon2}, we get 
\bea
&&(a-\epsilon)\int_{\cR_{{III}_A}}J^N_\mu n^\mu_{\cC_{v}}\dV_{\cR_{{III}_A}}\nonumber \\
&&
\le C\left( \int_{\cC_{v_0}(u_0(r^*),\infty)}J^N_{\mu}n^\mu_{\cC_{v}}\dV_{\cC_{v}}+ \int_{ \Sigma_{\tau_0} } J_{\mu}^{N}(\psi) n^{\mu}_{\Sigma_{\tau}}\dV_{ \Sigma_{\tau_0} }\right)
\eea
which combined with Proposition \eqref{energy-conservation} gives 
\bea
(a-\epsilon)\int_{\cR_{{III}_A}}J^N_\mu n^\mu_{\cC_{v}}\dV_{\cR_{{III}_A}}\le C \int_{ \Sigma_{\tau_0} } J_{\mu}^{N}(\psi) n^{\mu}_{\Sigma_{\tau}}\dV_{ \Sigma_{\tau_0} }.
\eea
\end{proof}


\subsection{Far Away Estimates }\lb{away}
{
In the following we will derive estimates in a far away region in close analogy to considerations in \cite[Section 4]{rp} 
}
\begin{thm}\label{far}
Let $N$ be the vector field multiplier given by Proposition \eqref{mi} and $\psi$ be a solution to \eqref{wave_psi} with compactly supported data. 
Then we have in the region {$\cR_{IV_C}:=\{(u,r)\in\cR_{IV}  \text{ s.t.}  r\ge w^*$, $u_0\le u\le u_1\}$} the following estimate 
\bea
\lb{div_D}
&&\int_{\cC_{u_1}(v_1(w^*),\infty)}r ^p(\partial_v\phi)^2 \md v+ \int\limits_{\cR_{IV_C}} p r^{(p-1)}(\partial_v\phi)^2 \md v\md u \nonumber \\
&&=\int_{\cC_{u_0}(v_0(w^*),\infty)}r^p (\partial_v\phi)^2 \md v+\left|\int\limits_{\tau_0}^{\tau_1} 2{w^*}^p (\partial_v\phi)^2 \md \tau\right|
\eea
for all $p>0$ where $\cR_{IV}:=\left\{D^+(\Sigma_{\tau_0}^{\tilde{r}{w}_{\star}})\cap D^{-}(\Sigma_{\tau_1}^{\tilde{r}{w}_{\star}})\right\}$ for $w_{\star}$ defined as below. See Figure \ref{far_away}.
\end{thm}
{\begin{figure}[!ht]
\centering
\includegraphics[width=0.45 \textwidth]{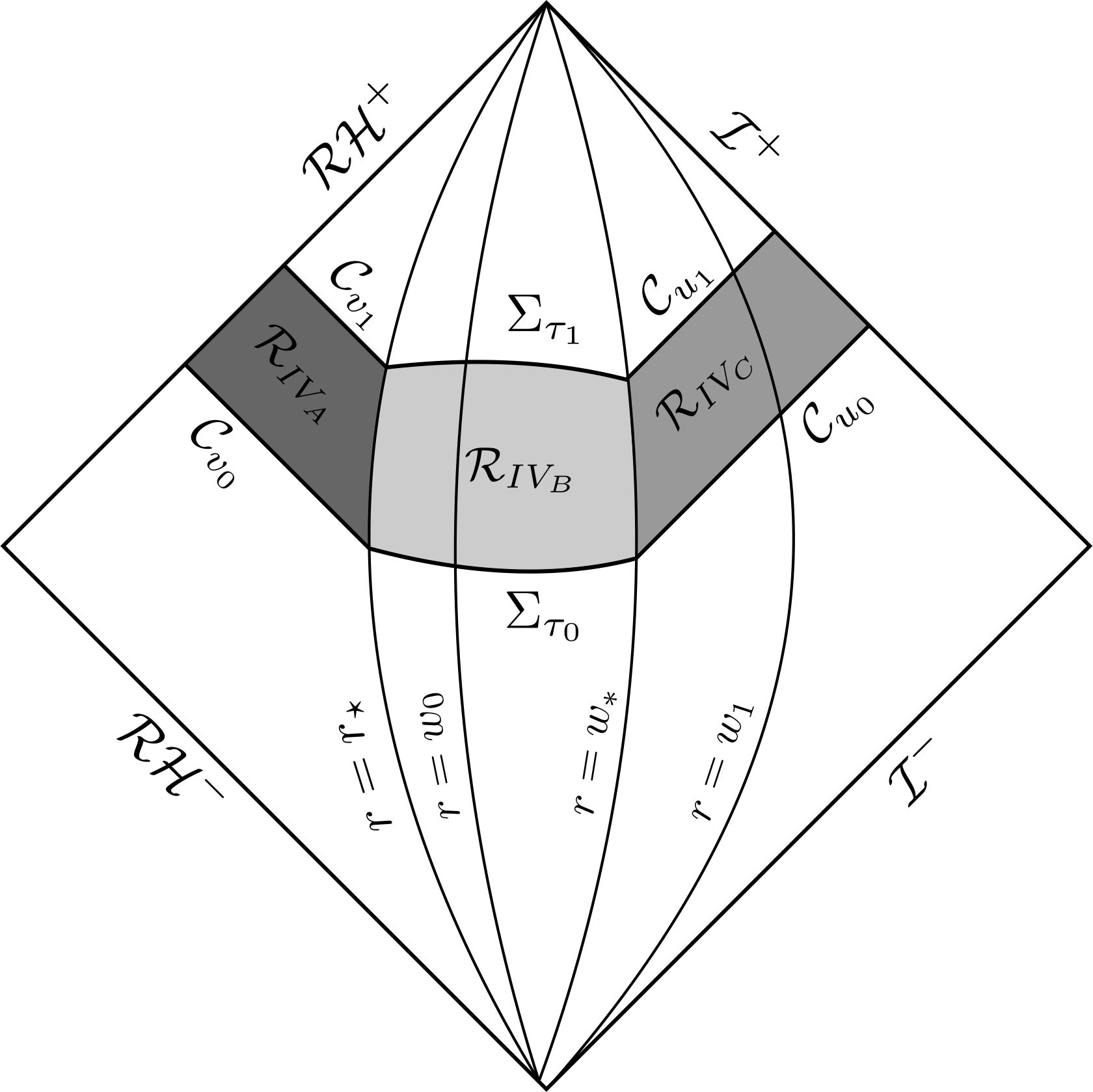}
\caption[]{Region $\cR_{{IV}}$ depicted with the far away region $\cR_{IV_C}$.}
\label{far_away}\end{figure}}
Applying the divergence theorem using the multiplier Eq.\eqref{Vvectorfield} in the entire region $\cR_{{IV}_C}$ and noting that the energy flux along $\cI^+$ does not contribute the given multiplier, leads to
\bea
\lb{entire}
&&\int_{\cC_{u_1}(v_1(w^*),\infty)}J^V_{\mu}n^{\mu}_{u=u_1} \dV_{\cC_{u_1}}+\int\limits_{\cR_{{IV}_C}} K^V \dV_{\cR_{{IV}_C}} \nonumber \\
&&
\leq\int_{\cC_{u_0}(v_0(w^*),\infty)}J^V_{\mu}n^{\mu}_{\cC_{u_0}} \dV_{\cC_{u_0}} + \left| \int\limits^{\tau_1}_{\tau_0}J^V_{\mu}n^{\mu}_{r=R}\dV_{r}\right|.
\eea
Using \eqref{JV}, \eqref{V} and \eqref{JV_urp}, we have
\bea
&&\int_{\cC_{u_1}(v_1(w^*),\infty)}2 r^{p} e^{-a(v-u)} (\partial_v\phi)^2 \dV_{\cC_{u_1}}\nonumber\\
&&
+ \int\limits_{\cR_{{IV}_C}} pr^{(p-1)}e^{-a(v-u)}(\partial_v\phi)^2 \dV_{\cR_{{IV}_C}} \nonumber\\
&\leq&\int_{\cC_{u_0}(v_0(w^*),\infty)}2 r^{p} e^{-a(v-u)} (\partial_v\phi)^2 \dV_{\cC_{u_0}}+\left|\int\limits_{\tau_0}^{\tau_1} 2{R^pe^{-aR}}(\partial_v\phi)^2 \dV_{r}\right|
\eea
which using \eqref{null_normal_u}, \eqref{normal_r} and \eqref{vol_ele} leads to the desired result \eqref{div_D}.

\subsection{Energy decay for solutions with compactly supported initial data}

\begin{thm}\label{main}
Let $\{\Sigma^{\tilde{r},r^*}_\tau\}_\tau$ be a folitiation such that $\tilde{r}<r_0$ and $r^*>0>r_0$ as in Theorem \ref{decaynear} and Theorem \ref{far},  $N$ be the vector field multiplier given by Proposition \eqref{mi} and $\psi$ be a solution to \eqref{wave_psi} with compactly supported data.  Then the following statements hold:
\be
\lb{173}
\int_{\Sigma^{\tilde{r},r^*}_\tau}J_{\mu}^{N}(\psi) n^{\mu}_{\Sigma^{\tilde{r},r^*}_\tau}\dV_{\Sigma^{\tilde{r},r^*}_\tau}\le \frac{C}{\tau}\int_{ \Sigma_0 } J_{\mu}^{N}(\psi) n^{\mu}_{\tau}\dV_{\Sigma_0}
\ee 
\bea
\label{square}
&&\int_{\Sigma^{\tilde{r},r^*}_\tau}J_{\mu}^{N}(\psi) n^{\mu}_{\tau}\dV_{\Sigma^{\tilde{r},r^*}_\tau}\nonumber\\
&&\le  C\frac{C}{\tau^n}\left(\int_{\cC_{u_0}(v_0(\tilde{r}),\infty)}r^n (\partial_v\phi)^2 \md v+ C \int_{ \Sigma_{\tau_0} } J_{\mu}^{N}(\psi) n^{\mu}_{\tau}\dV_{ \Sigma_{\tau_0}}\right),
\eea 
for all $n\in\mathbb{N}$.
\end{thm}
\begin{proof}
Using Theorem \ref{far} for $p=1$ we have 
\bea
\label{p1}
&&\int_{\cC_{u_1}(v_1(\tilde{r}),\infty)}r (\partial_v\psi)^2 \md v+ \int_{\cR_{{IV}_C}} (\partial_v\psi)^2 \md v\md u \nonumber \\
&&=\int_{\cC_{u_0}(v_0(\tilde{r}),\infty)}r (\partial_v\psi)^2 \md v+\left|\int_{\tau_0}^{\tau_1} 2R (\partial_v\psi)^2 \md\tau\right|,
\eea
and  notice that in this region we have
\be
\int_{\cR_{{IV}_C}} J_{\mu}^{N}(\psi) n^{\mu}_{C_u} \dV_{C_u}=\int_{\cR_{{IV}_C}} J_{\mu}^{\partial_\tau}(\psi) n^{\mu}_{C_u} \dV_{C_u} =\frac{1}{2}  \int_{\cR_{{IV}_C}} (\partial_v\psi)^2 \md v\md u.
\ee
Also, 
\be
\lb{176}
\left|\int_{\tau_0}^{\tau_1} 2r^* (\partial_v\psi)^2 d\tau\right| \le C \int_{ \Sigma_{\tau_0} } J_{\mu}^{N}(\psi) n^{\mu}_{\tau}\dV_{ \Sigma_0 }
\ee


Putting this together in Eq. (\ref{p1}) and ignoring a positive contribution from the left side we get
\be
\int_{\cR_{{IV}_C}}J_{\mu}^{\partial_\tau}(\psi) n^{\mu}_{\tau}\dV_{\cR_{{IV}_C}}\le C \left (\int_{\cC_{u_0}(v_0(\tilde{r}),\infty)}r (\partial_v\psi)^2 dv+\int_{\Sigma_{\tau_0}} J_{\mu}^{N}(\psi) n^{\mu}_{\tau}\dV_{\Sigma_{\tau_0}}\right).
\ee 
Using the estimates \eqref{funf}, \eqref{near}, we obtain 
\bea
&\int_{\cR_{{IV}_C}}J_{\mu}^{N}(\psi) n^{\mu}_{\tau}\dV_{\cR_{{IV}_C}}+ \int_{(\cR_{IV_A}\cup\cR_{{IV}_B})}J_{\mu}^{N}(\psi) n^{\mu}_{\tau}\dV_{(\cR_{IV_A}\cup\cR_{{IV}_B})} \nonumber \\
&\le\int_{\cR_{{IV}_C}}J_{\mu}^{N}(\psi) n^{\mu}_{\tau}\dV_{\cR_{{IV}_C}}+ \int_{\cR_{III}}J_{\mu}^{N}(\psi) n^{\mu}_{\tau}\dV_{\cR_{III}} \nonumber \\
&\le C \left (\int_{\cC_{u_0}(v_0(\tilde{r}),\infty)}r (\partial_v\psi)^2 dv+\int_{\Sigma_{\tau_0}} J_{\mu}^{\partial_\tau}(\psi) n^{\mu}_{\tau}\dV_{\Sigma}\right),
\eea 
combining  the coarea formula with Proposition \ref{energy-conservation}  gives the estimate \eqref{173}.
Moreover,  from Theorem \ref{far} we have  in the region \mbox{$\cR_{{IV}_C}$} the following estimate 
\bea
&&\int_{\cC_{u_1}(v_1(\tilde{r}),\infty)}r ^p(\partial_v\phi)^2 \md v+ \int\limits_{\cR_{{IV}_C}} p r^{(p-1)}(\partial_v\phi)^2 \md v\md u \nonumber \\
&& \qquad
\leq\int_{\cC_{u_0}(v_0(\tilde{r}),\infty)}r^p (\partial_v\phi)^2 \md v+ C \int_{ \Sigma_{\tau_0} } J_{\mu}^{N}(\psi) n^{\mu}_{\tau}\dV_{ \Sigma_{\tau_0} }
\eea
This implies we can find a dyadic sequence $\tau_n\rightarrow \infty$ with the property that 
\bea
\label{p2}
&&\int_{\cC_{u_{\tau_n}}(v_(w_{\tau_n})^*),\infty)} p r^{(p-1)}(\partial_v\phi)^2 \md v \nonumber \\
&&
\leq C\tau_n^{-1}\left(\int_{\cC_{u_0}(v_0(\tilde{r}),\infty)}r^p (\partial_v\phi)^2 \md v+ C \int_{ \Sigma_{\tau_0} } J_{\mu}^{N}(\psi) n^{\mu}_{\tau}\dV_{ \Sigma_{\tau_0}}\right)
\eea
We now apply Eq. \eqref{div_D} with $p-1$ in place of $p$ to obtain 
\bea
&&\int\limits_{\cC_{u_{\tau_n}}(v_(w_{\tau_n}^*),\infty)} r ^{p-1}(\partial_v\phi)^2 \md v+ \int\limits_{{\cR_{{IV}_C}}^{\tau_n}_{\tau_{n-1}}} (p-1) r^{(p-2)}(\partial_v\phi)^2 \md v\md u\nonumber \\
&&\stackrel{\eqref{176}}{\le}\int_{\cC_{u_{\tau_{n-1}}}(v_(w_{\tau_{n-1}}^*),\infty)}r^{p-1} (\partial_v\phi)^2 \md v+  C \int_{\Sigma^{\tilde{r},r^*}_{\tau_{n-1}}} J_{\mu}^{N}(\psi) n^{\mu}_{\tau}\dV_{\Sigma^{\tilde{r},r^*}_{\tau_{n-1}}} \nonumber \\
&&\stackrel{\eqref{p2}}{\le} C\tau_n^{-1}\left(\int_{\cC_{u_0}(v_0(\tilde{r}),\infty)}r^p (\partial_v\phi)^2 \md v+ C \int_{ \Sigma_{\tau_0} } J_{\mu}^{N}(\psi) n^{\mu}_{\tau}\dV_{ \Sigma_{\tau_0}}\right)\nonumber \\
&& \qquad
+ C  \int_{\Sigma^{\tilde{r},r^*}_{\tau_{n-1}}} J_{\mu}^{N}(\psi) n^{\mu}_{\tau}\dV_{\Sigma^{\tilde{r},r^*}_{\tau_{n-1}}}
\eea
in the region ${\cR_{{IV}_C}}^{\tau_n}_{\tau_{n-1}}$. 
For $p=2$ and adding a multiple of the ILED \eqref{funf}, \eqref{near} we get 
\bea
&&\int_{{\cR_{{IV}_C}}^{\tau_n}_{\tau_{n-1}}}J_{\mu}^{N}(\psi) n^{\mu}_{\tau}\dV_{{\cR_{{IV}_C}}^{\tau_n}_{\tau_{n-1}}}\nonumber \\
&&+ \int_{(\cR_{IV_{A}}{^{\tau_n}_{\tau_{n-1}}}\cup\cR_{{IV}_B}{^{\tau_n}_{\tau_{n-1}}})}J_{\mu}^{N}(\psi) n^{\mu}_{\tau}\dV_{(\cR_{IV_{A}}{^{\tau_n}_{\tau_{n-1}}}\cup\cR_{{IV}_B}{^{\tau_n}_{\tau_{n-1}}})}\nonumber \\
&&\le  C\tau_n^{-1}\left(\int_{\cC_{u_0}(v_0(\tilde{r}),\infty)}r (\partial_v\phi)^2 \md v+ C \int_{ \Sigma_{\tau_0} } J_{\mu}^{N}(\psi) n^{\mu}_{\tau}\dV_{ \Sigma_{\tau_0}}\right)\nonumber \\
&&+ C \int_{\Sigma^{\tilde{r},r^*}_{\tau_{n-1}}} J_{\mu}^{N}(\psi) n^{\mu}_{\tau}\dV_{\Sigma^{\tilde{r},r^*}_{\tau_{n-1}}}
\eea
Using  Eq.\eqref{173}, the properties of the dyadic sequence and the energy boundness given by Proposition \ref{energy-conservation}  we get 
\bea
&&\int_{\Sigma^{\tilde{r},r^*}_{\tau}}J_{\mu}^{N}(\psi) n^{\mu}_{\tau}\dV_{\Sigma^{\tilde{r},r^*}_{\tau}}\nonumber \\
&&\le  C\frac{C}{\tau^2}\left(\int_{\cC_{u_0}(v_0(\tilde{r}),\infty)}r^2 (\partial_v\phi)^2 dv+ C \int_{ \Sigma_{\tau_0} } J_{\mu}^{N}(\psi) n^{\mu}_{\tau}\dV_{ \Sigma_{\tau_0}}\right)
\eea 
for any $\tau\ge \tau_0$.
Using  Theorem \ref{far} with $p=3$ , another dyadic sequence,  the mean value theorem, \ref{energy-conservation} and Eq. \eqref{p2} we obtain the estimate 
\bea
&&\int_{\Sigma_{\tau_n}}J_{\mu}^{\partial_\tau}(\psi) n^{\mu}_{\tau}\dV\nonumber \\
&& \le  \frac{C}{\tau^2}\left(\frac{C}{{\tau_n}}\left(\int_{\cC_{u_0}(v_0(\tilde{r}),\infty)}r^3 (\partial_v\phi)^2 \md v+ C \int_{ \Sigma_{\tau_0} } J_{\mu}^{\partial_\tau}(\psi) n^{\mu}_{\tau}\dV_{ \Sigma_{\tau_0}}\right)\right.\nonumber \\
&&
\qquad \left.+ C \int_{\Sigma^{\tilde{r},r^*}_{\tau_{n-1}}} J_{\mu}^{N}(\psi) n^{\mu}_{\tau}\dV_{\Sigma^{\tilde{r},r^*}_{\tau_{n-1}}}\right)\nonumber \\
&&\le  \frac{C}{\tau^2}\left(\frac{C}{_{\tau_n}}\left(\int_{\cC_{u_0}(v_0(\tilde{r}),\infty)}r^3 (\partial_v\phi)^2 \md v+ C \int_{ \Sigma_{\tau_0} } J_{\mu}^{\partial_\tau}(\psi) n^{\mu}_{\tau}\dV_{ \Sigma_{\tau_0}}\right)
\right.\nonumber \\
&&
\qquad \left.
+  \frac{C}{\tau_n}\int_{ \Sigma_{\tau_0} } J_{\mu}^{\partial_\tau}(\psi) n^{\mu}_{\tau}\dV_{ \Sigma_{\tau_0}}\right) \nonumber \\
&&\le  \frac{C}{\tau \tau_n}\left(\int_{\cC_{u_0}(v_0(\tilde{r}),\infty)}r^3 (\partial_v\phi)^2 \md v+  \int_{ \Sigma_{\tau_0} } J_{\mu}^{\partial_\tau}(\psi) n^{\mu}_{\tau}\dV_{ \Sigma_{\tau_0}}\right)
\eea
Since  we can repeat the argument for all $p\in\mathbb{N}$, we obtain 
\be
\int_{\Sigma^{\tilde{r},r^*}_\tau}J_{\mu}^{\partial_\tau}(\psi) n^{\mu}_{\tau}\dV_{\Sigma_\tau}\le  \frac{C}{\tau^n}\left(\int_{\cC_{u_0}(v_0(\tilde{r}),\infty)}r^n (\partial_v\phi)^2 dv+  \int_{ \Sigma_{\tau_0} } J_{\mu}^{\partial_\tau}(\psi) n^{\mu}_{\tau}\dV_{ \Sigma_{\tau_0}}\right)
\ee
for all $n\in \mathbb{N}$.
\end{proof}
{

\section{Discussion}
\noindent
Due to the equivalence principle, the energy decay of a family of uniformly accelerated observers must share similar properties with the energy decay of a family of stationary observers in a gravitational field. In this section, we focus on the comparison between Rindler spacetime and Schwarzschild spacetime. This comparison follows from the observation that in Schwarzschild the metric is given by
 \begin{align*}
ds^2&=-(1-\frac{2m}{r})dt^2+(1-\frac{2m}{r})^{-1}dr^2+r^2d\Omega^2\\
&=-\xi^2dt^2+(1-\xi^2)^{-4}d\xi^2+r^2(\xi)d\Omega^2 \\
&=-{{\xi^2 dt^2+(1+4\xi^2+...)d\xi^2}}+r^2(\xi)d\Omega^2
\end{align*}
\noindent
for $\xi^2:=1-\frac{2m}{r}<1$ where $t$ is proportional to the proper time of static observers and $\xi$ small describes the near horizon geometry, while the Rindler metric is given by
$$ds^2=-\xi^2 d\tau^2+d\xi^2+dy^2+dz^2$$
\noindent
for $0<\xi<\infty$ where $\tau$ is proportional to the  proper time of uniformly accelerated observer and uniformly accelerated observers remain at fixed $\xi$. 

\noindent

\subsection{Comparison with Schwarzschild spacetime.}
We discuss three separate regions as shown in the Figure below.

\begin{figure}[h]
\centering\includegraphics[width=4.7cm]{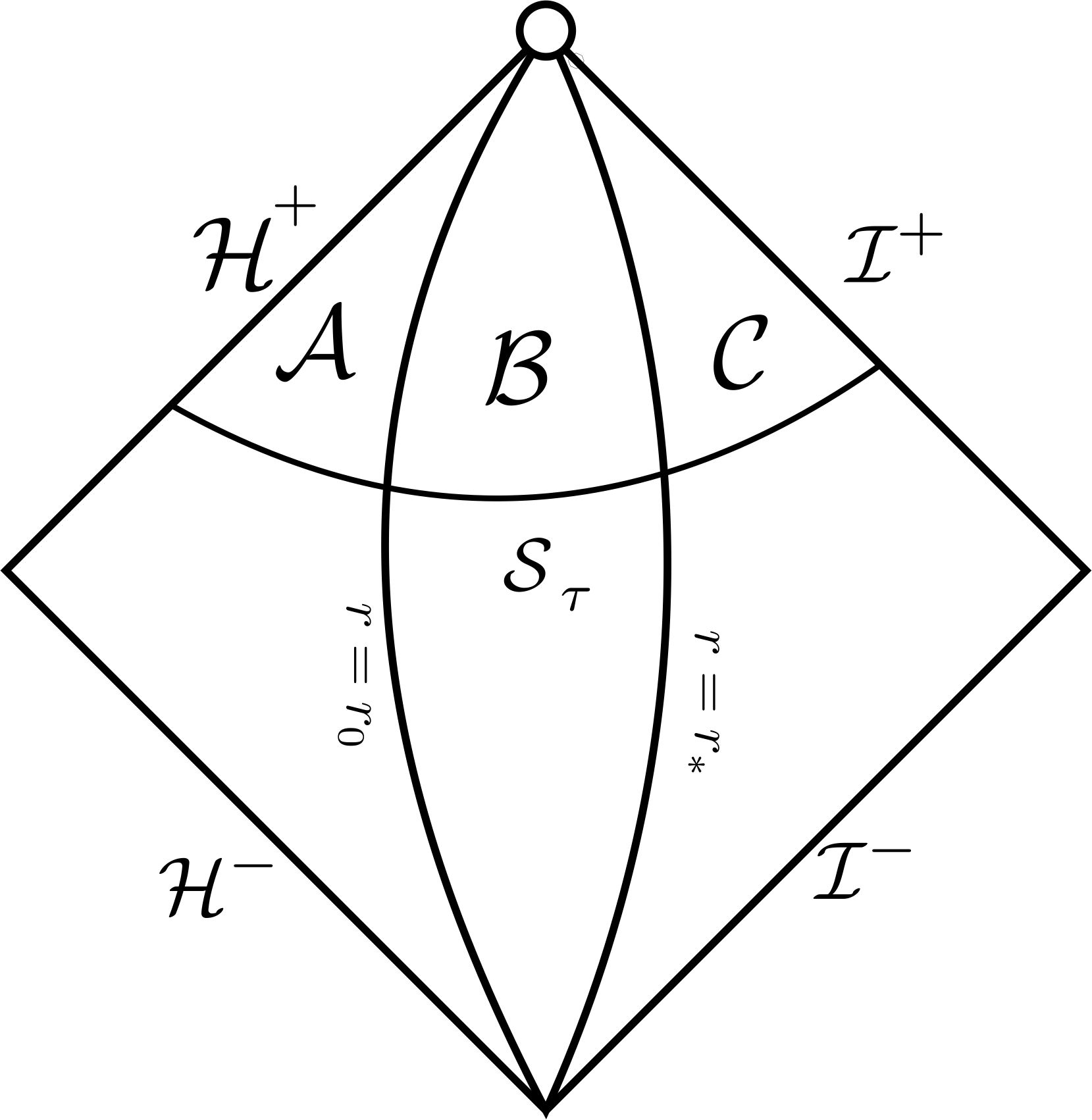}\quad\includegraphics[width=4.7cm]{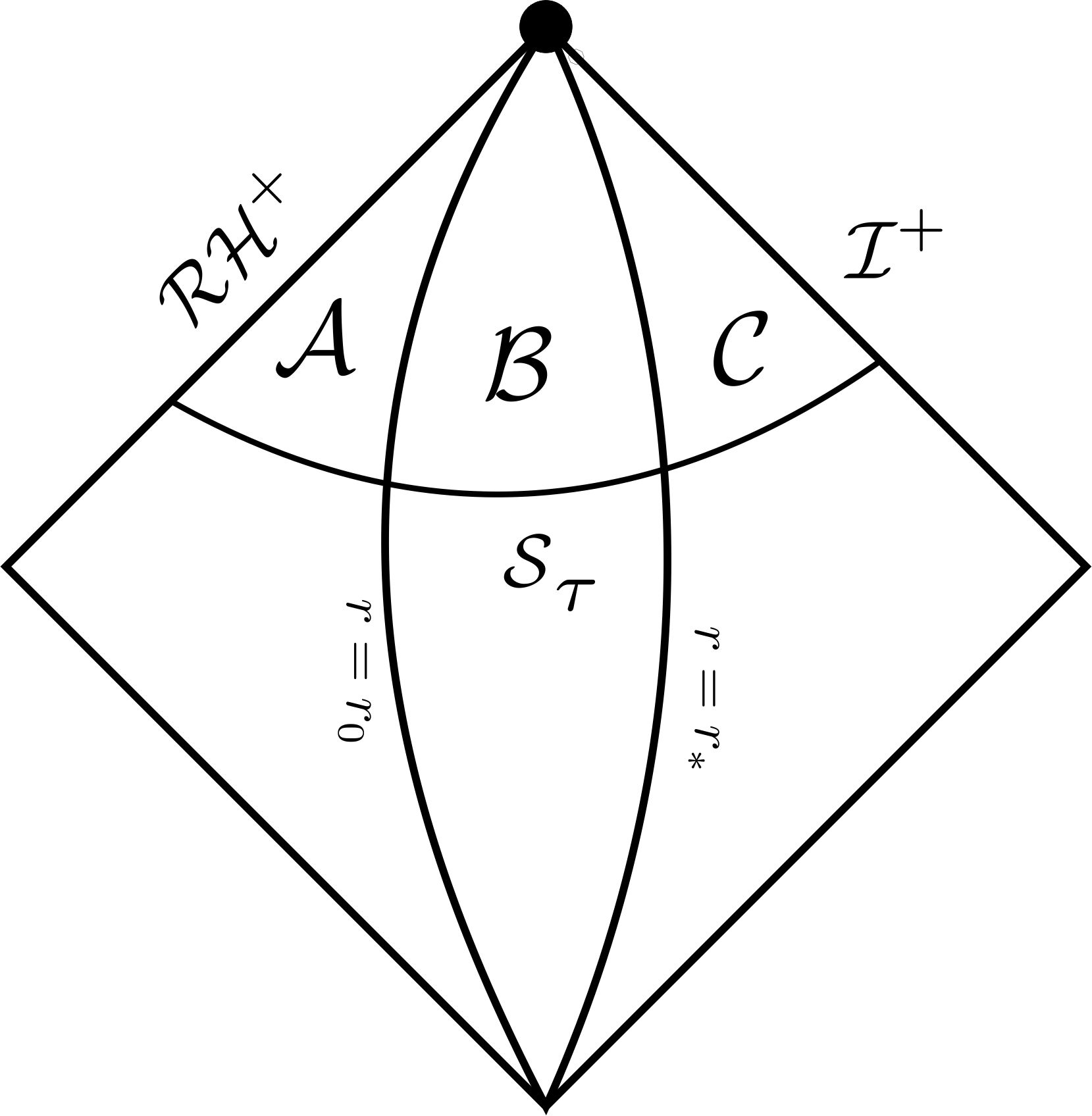}
\caption{Left: Schwarzschild spactime divided into three regions by the radial coordinate $r$. Right: Rindler spacetime divided in three regions by the coordinate $r$ described in \eqref{r-t-coords}.}
\end{figure} 

\subsection{Region ${\cal{B}}$}
This region corresponds to compact regions away from the horizon. The analysis of trapping effects in Schwarzschild, due to the photon sphere located at $r=3M$,  are treated here. At the level of the estimates one needs to lose derivatives in order to obtain a Integrated Local Energy Decay estimate \cite[Section 4.1]{lectures}. Since in Rindler there is no such trapping behaviour, we obtain the Integrated Local  Energy Decay estimate without loss of derivatives (Theorem \ref{decaymedium} Eq.\eqref{funf}). Notice that the situation also differs to what happens in Minkowski spacetime. There for compact regions there exists a finite $t$ such that the energy flux defined with respect to the Killing vector $\partial_t$  is zero. Although, Rindler spacetime is a subset of Minkowski spacetime, the energy flux of the solutions with respect to $\partial_\tau$  in any compact region does not necessarily vanish for any finite $\tau$. 

\subsection{Region ${\cal{C}}$}

The far away region is usually analysed using conformal compactifications. In the $1+1$ Rindler spacetime, by inspection one can show that part of the Minkowski null infinity region coincide with the null infinity region in Rindler \cite{semay, socolov}. However, the conformal compactification of Rindler can be done in several ways e.g. by looking at the corresponding region in the conformal compactification of Minkowski spacetime or using a conformal transformation that sends the Rindler wedge into a casual diamond in Minkowski spacetime \cite[Appendix E]{jacobson}, \cite[V.4.2]{haag}. 
 It is in this region where we are able to obtain better estimates for the $r^p$ vector field compared to the Schwarzschild case. This may be a consequence of the $1+1$ case. In the higher dimensional case the transversal direction may constraint the range of possible values of $p$. Furthermore, in higher dimensional Rindler spacetimes there is also the ambiguity of what should be the suitable analogue to the usual spherical coordinates.
 In addition, the asymptotic symmetries at infinity in Rindler spacetime differ from the asymptotic symmetries in the  asymptotically flat case \cite{asymptotic}. Finally, unlike, Schwarzschild and Minkowski spacetime where timelike infinity (denoted by the white dot)  is defined using timelike geodesics,  in the Rindler case, the `timelike infinity' (denoted by the black dot) is defined using the uniformly accelerated observers.

\subsection{Region ${\cal{A}}$}

In the near horizon Rindler region, the existence of a bifurcated Killing horizon with positive surface gravity allows one to construct the redshift vector field $Y$ and the local observers vector field $N$ as in the Schwarzschild case. However, in Schwarzschild spacetime there is a compact bifurcation surface which is not the case in higher dimensional Rindler spacetimes where there is a bifurcation hyperplane. Nevertheless, the existence of the redshift vector in both spacetimes suggests further analogies. For example, in Schwarzschild, the redshift effect in  scattering constructions from the future turns into a blueshift. This statement can be formulated as the non-invertibility of forward scattering maps. To be precise, one can define in Schwarzschild spacetime, the maps $$ E^{X}_{\Sigma_{\tau_0}}\to E^{X}_{{\cal{H}}^+}\oplus E^{X}_{{\cal{I}}^+}$$ for $X=\partial_{t}$ and  $X=N$ where $$E^X_{{\cal{M}}}(\psi):=\int_{{\cal{M}}}T_{\mu\nu}(\psi)X^{\mu}n^{\nu}_{{\cal{M}}}\text{dV}_{{\cal{M}}}=\int_{{\cal{M}}}J^X_{\nu}(\psi)n^{\nu}_{{\cal{M}}}\text{dV}_{{\cal{M}}}$$ and $\Sigma_{\tau_0}$ is a suitable initial Cauchy hypersurface. The existence of the forward scattering maps can be done using decay estimates combined either with  the existence of radiation fields  \cite{time, mos} or a suitable conformal compactification \cite{nicolas}. Moreover, it was shown in \cite[Theorem 4.1]{time} that for $X=\partial_t$ there is an inverse map  $$E^{\partial_t}_{{\cal{H}}^+}\oplus E^{\partial_t}_{{\cal{I}}^+}\to E^{\partial_t}_{\Sigma_{\tau_0}}$$ while for $X=N$ this is not possible. The non-invertibility is  a consequence of the time reversed redshift at the horizon.  

To further, give evidence that this is also the case in Rindler,  we will use the characterisation of the energy of Gaussian Beams on Lorentzian manifolds.
Notice that by \cite[Theorem 5.1]{gaussianbeams} there exists a smooth solution $v$ of the wave equation with $E^{X}_{t=0}(v) = -g(X, \dot{\gamma}(\gamma(0)))$ such that
\[
|E^{X}_{t=T_0,\mathcal{N} \cap \Sigma_\tau}(v) - (-g(X, \dot{\gamma}))_{\operatorname{Im}\gamma \cap T=T_0}| < \mu
\]
for all $0 \leq t \leq T_1$, where  $t$ is the Minkowski time coordinate map. Furthermore, $\gamma : [0, S) \to M$ is the affinely parametrized generator of $\cal{RH}^+$, $X$  a timelike, future-directed vector field and $\mathcal{N}$ a neighbourhood of $\gamma$.

Explicitly, we have $\gamma(s)=(s,s)$ and $-g(N,\dot{\gamma})=e^{-as}$ which means that the energy of the corresponding Gaussian beam decays exponentially. This is  a direct manifestation of red-shift effect at the Rindler Horizon. Moreover, since all hypothesis are satisfied, an analogous proposition to \cite[Proposition 6.3]{gaussianbeams} holds. Explicitly,

\begin{prop}
 For every $\mu > 0$ and every $t > 0$, there exists a smooth solution $v$ to the wave equation with $E^N_{t}(v) = 1$ and $\int_{\cal{RH}^+}J_\nu^N n^\nu dV<\mu$ 
which satisfies
\[
E^N_{0}(v) \geq e^{at} - \mu,
\]
where $a$ is the surface gravity at the Rindler Horizon.
\end{prop}

The proposition demonstrates that for every $t_0 > 0$, initial data can be specified for the mixed characteristic initial value problem on $\mathcal{RH}^+ \cup \{(t,x):t=t_0,x>t_0\}$ such that the total initial energy equals one. Meanwhile, by solving the equation backwards, the energy of the resulting solution approaches approximately $e^{a t_0}$ on $t=0$.

\subsection{The principle of equivalence.}
The equivalence principle was initially stated  by Einstein as the assumption that there is  a  ``complete physical equivalence of a gravitational field and a corresponding acceleration of the reference system" \cite{einstein}. It is a geometric
consequence of General Relativity and was a crucial guiding principle during the initial developing
stages of the theory. Moreover, it has been experimentally verified with precision up to the atomic level, see for example \cite{will} or the contribution of Shapiro in \cite{shapiro}. It is worth noting that the analysis of the equivalence principle at the quantum level has been explored, as evidenced by  references \cite{giua},\cite{giub} and \cite{det}. Precise definitions of the equivalence principle are crucial for this analysis. For a detailed discussion on the various formulations of this principle, along with a historical overview, we refer the reader to \cite{dennis}.

In what follows, we aim to explore the extent to which classical wave equations render acceleration and gravity indistinguishable, and what distinct signatures might allow us to differentiate between them. As previously discussed, Rindler spacetime, which represents flat spacetime from the perspective of uniformly accelerated observers, exhibits similarities to, yet important differences from, the gravitational field produced by a spherically symmetric static mass, as described by the Schwarzschild metric. These differences and similarities provide key insights into the underlying physics of acceleration and gravity.

It is clear from the above comparison that in Region $\mathcal{A}$, which is close to the horizon, the behavior of wave equations is similar; for instance, there is a degeneration of energy with respect to static observers in the Schwarzschild case and with respect to accelerated observers in Rindler. Additionally, there is an exponential energy decay with respect to local observers, as indicated by Gaussian beams. This similarity is due to the presence of a Killing horizon with associated surface gravity. It is noteworthy that these effects occur regardless of whether one spacetime is flat and the other can be possibly strongly curved.

In the intermediate region $\cal{B}$, there are distinguishable effects between the two spacetimes. The behavior in Schwarzschild spacetime, particularly trapping phenomena, distinctively influences local energy decay estimates. From a physical perspective, the photon sphere acts as an obstruction to decay; however, the geodesics in this region are unstable. Nevertheless, in principle, an observer located near the photon sphere would not notice energy decay, in contrast to what is observed in Rindler spacetime, where such obstructions are absent. Therefore, one could distinguish between accelerated frames and gravitational fields by measuring the decay rates of energy in the associated field, provided one is sufficiently far away from the horizon and close to the photon sphere. 

In the distant region $\mathcal{C}$, Schwarzschild and Rindler spacetimes exhibit a notable similarity concerning the interaction between the horizon, null infinity, and their associated scattering constructions. In Schwarzschild spacetime, an early-time energy blow-up occurs for decaying data observed at null infinity and vanishing at the event horizon, a phenomenon resulting from the redshift transforming into a blueshift. Above we have shown, that a similar blow-up is likely to occur also in Rindler. A parallel situation occurs in quantum physics. For instance, in Schwarzschild, the energy-momentum tensor of the Boulware vacuum diverges at the event horizon. Similarly, the energy-momentum tensor of the Rindler vacuum also exhibits a singularity at the horizon, indicating a similar kind of blow-up.

We have presented some results that shows the extent to which the principle of equivalence holds at the level of massless wave equations. This issue is subtle, in particular in relation with locality and the role of curvature. Our discussion highlights that certain global behaviors—specifically related to the energy of waves observed by accelerated observers in flat spacetime compared to that observed by static observers in a gravitational field—exhibit similarities. Moreover, close to the horizon, the redshift effect on the energy decay is similar, even in the presence of strong curvature in the gravitational case. Therefore, the validity of this principle depends critically on the specific statement one aims to prove.

}
\section*{Acknowledgements}  
The authors would like to thank the Max Planck Institute of Mathematics in Bonn for its hospitality during the writing of this  paper.
We also thank the  Oberwolfach Institute, Leibniz University Hannover and CAMGSD, IST Lisboa for additional support.
Special thanks to Pedro Gir\~ao for valuable comments on the manuscript. Further, we also benefited from discussions with José Natário and Michael Gruber. 
This work was partially supported by FCT/Portugal through CAMGSD, IST-ID , projects UIDB/04459/2020
and UIDP/04459/2020, by FCT/Portugal and through the FCT fellowship CEECIND/00936/2018 (A.F.). 

We have no conflicts of interest to disclose.

\begin{appendix}

\section{Regular double-null coordinates}
\lb{reg_dn}
In double null coordinates $(U,v)$ which are regular at $\cR\cH^+$, we use the transformation
\bea
\lb{U-coord}
u(U)&=&-\frac{1}{a}\ln |a U|, \qquad U=-\frac{1}{a}e^{-au},\\
\mbox{with} \qquad \md u&=&\frac{\partial u}{\partial U} \md U=-\frac{1}{aU}\md U={e^{au}}\md U,
\eea
and obtain the metric
\bea
\lb{doublmetricreg}
\md s^2=-{e^{av}}\md U \md v .
\eea
Defining the null cordinate $V$ as
\bea
\lb{V-coord}
v(V)&=&\frac{1}{a}\ln (a V), \qquad V=\frac{1}{a}e^{av},\\
\mbox{with} \qquad \md v&=&\frac{\partial v}{\partial V} \md V=\frac{1}{aV}\md V={e^{-av}}\md V,
\eea
we obtain in the $(u,V)$ coordinate system the metric
\bea
\lb{doublmetricregV}
\md s^2=-{e^{-au}}\md u \md V .
\eea
Hence, in the $(U,V)$ coordinate system we get the familiar Minkowski form of the metric
\bea
\lb{doublmetricreg2}
\md s^2=-\md U \md V ,
\eea
and the relationship between laboratory coordinates and  regular double null coordinates is given by
\bea
U=t-x, \quad \mbox{and} \quad V=t+x,
\eea
restricted to the subdomain $x>|t|$ yielding the range $-\infty<U<0$ and $0<V<\infty$.

\section{The $K$-current}
In order to calculate the scalar-currents for the bulk term according to \eqref{K} we fist derive the components of the deformation tensor 
\bea
\lb{def_tensor}
(\pi^X)^{\mu\nu}=\frac{1}{2}(g^{\mu\lambda}\partial_\lambda X^{\nu}+g^{\nu\lambda}\partial_\lambda X^{\mu}+g^{\mu\lambda}g^{\nu\delta}g_{\lambda\delta,\rho}X^{\rho})
\eea
in different coordinates.
In $(r, \tau, y, z)$-coordinates with the metric \eqref{r-tau-metric} and for the arbitrary vector field $X=X^r(r, \tau)\partial_r+X^{\tau}(r, \tau)\partial_{\tau}$ we obtain
\bea
\lb{def_components_rtau}
(\pi^X)^{\tau\tau}&=&{e^{-2ar}}[\partial_{\tau} X^{\tau}+aX^{r}],\\
(\pi^X)^{{r}{r}}&=&{e^{-2ar}}[\partial_{r} X^{r}+aX^{r}],\\
(\pi^X)^{\tau {r}}&=&\frac{e^{-2ar}}{2}[-\partial_{\tau} X^{r}+\partial_{r}X^{\tau}],
\eea
Further, with  \eqref{energymomentum} the energy-momentum tensor yields
\bea
\lb{tensor_rtau}
T_{{\tau}{\tau}}&=&\frac{1}{2}[(\partial_{\tau}\psi)^2+(\partial_{r}\psi)^2],\\
T_{{r}{r}}&=&\frac{1}{2}[(\partial_{\tau}\psi)^2+(\partial_{r}\psi)^2],\\
T_{{\tau}r}&=&(\partial_{\tau}\psi\partial_{r}\psi),
\eea
we obtain the scalar current
\bea
\lb{K_current_rtau}
K^X&=& \quad (\partial_{\tau} \psi)^2 \frac{e^{-2ar}}{2}\left[\partial_r X^r - \partial_{\tau} X^{\tau} \right]\nonumber \\
&& +(\partial_{r} \psi)^2 \frac{e^{-2ar}}{2}\left[\partial_r X^r - \partial_{\tau} X^{\tau}\right]\nonumber \\
&&+ (\partial_{\tau} \psi\partial_{r} \psi){e^{-2ar}}\left[\partial_r X^{\tau} - \partial_{\tau} X^{r}\right]
\eea
In $(u,v)$-coordinates with the metric \eqref{doublmetric} and for the arbitrary vector field {$X=X^u(u,v)\partial_u+X^v(u,v)\partial_v$} 
we get 
\bea
\lb{def_components_uv}
(\pi^X)^{uu}&=&-2e^{-a(v-u)}\partial_v X^u,\\
(\pi^X)^{vv}&=&-2e^{-a(v-u)}\partial_u X^v,\\
(\pi^X)^{uv}&=&-e^{-a(v-u)}[\partial_v X^v+\partial_u X^u]+ae^{-a(v-u)}[ X^u- X^v].
\eea
With \eqref{energymomentum} we get the following components for the energy-momentum tensor,
\bea
\lb{tensor_uv}
T_{uu}&=&(\partial_u\psi)^2,\\
T_{vv}&=&(\partial_v\psi)^2,\\
T_{uv}&=&0,
\eea
Multiplying the components then leads to 
\bea
\lb{K_current_uv}
K^X&=&-2e^{-a(v-u)}[(\partial_u\psi)^2\partial_vX^u+(\partial_v\psi)^2\partial_uX^v].
\eea

\section{The $J$-current}
\lb{J_current}

 For the timelike normal vector on a constant $\tau$ hypersurface we have
\bea
\lb{normal_tau2}
&n^\mu_{\tau=c}=e^{-ar}\partial_{\tau},\quad \mbox{with} \quad &\dV_{\tau}={e^{ar}}\md r,
\eea
so that we can derive the currents
\bea
\lb{JX_tau}
J^{X}_{\mu}n^\mu_{\tau=c}&=&\frac{X^{\tau}e^{-ar}}{2}\left[[(\partial_{\tau}\psi)^2+(\partial_{r}\psi)^2]\right]\nonumber \\
&&+X^re^{-ar}(\partial_r \psi \partial_{\tau} \psi),
\eea
\bea
\lb{JX_r}
J^{X}_{\mu}n^\mu_{r=c}&=&\frac{X^{r}e^{-ar}}{2}\left[(\partial_{\tau}\psi)^2+(\partial_{r}\psi)^2\right]\nonumber \\
&&+X^{\tau}e^{-ar}(\partial_r \psi \partial_{\tau} \psi),
\eea
with $X=X^r(r,\tau)\partial_r+X^{\tau}(r,\tau)\partial_{\tau}$.\\
A vector field multiplier $X^u(u,v)\partial_u+X^{v}(u,v)\partial_{v}$ leads to the currents
\bea
\lb{JX_u}
J^{X}_{\mu}n^\mu_{u=c}&=&X^v2e^{-a(v-u)}(\partial_{v}\psi)^2,\\
\lb{JX_v}
J^{X}_{\mu}n^\mu_{v=c}&=&X^u2e^{-a(v-u)}(\partial_{u}\psi)^2,\\
\lb{JX_tau_null}
J^{X}_{\mu}n^\mu_{\tau=c}&=&e^{-\frac{a(v-u)}{2}}\left[X^u(\partial_{u}\psi)^2+X^v(\partial_{v}\psi)^2\right]\nonumber.
\eea

\end{appendix}

\bibliographystyle{plain}
\bibliography{refs3}
\end{document}